\documentclass[a4paper,USenglish,cleveref,autoref,thm-restate,numberwithinsect,nameinlink]{lipics-v2021}

\usepackage[utf8]{inputenc}
\usepackage{amsmath}
\usepackage{amssymb}
\usepackage{enumerate}
\usepackage{dsfont}
\usepackage{thmtools}
\usepackage{mathtools}
\usepackage{graphicx}
\usepackage{multirow}
\usepackage{xcolor}
\usepackage{tikz}
\usetikzlibrary {arrows.meta} 
\usetikzlibrary{graphs}
\usepackage{placeins}
\usepackage{nicefrac}
\usepackage{tabularx}
\usepackage{algorithm2e}
\usepackage[T1]{fontenc}
\usepackage{etoolbox}
\usepackage{hyperref}
\usepackage{paralist}
\hypersetup{
	colorlinks,
	linkcolor={red!50!black},
	citecolor={blue!90!black!80},
	urlcolor={blue!80!black}
}
\usepackage{cleveref}

\usepackage[disable]{todonotes}
\nolinenumbers

\newtheorem{rr}{Reduction~Rule}{\upshape\itshape}{\upshape\rmfamily}

{\upshape\itshape}{\upshape\rmfamily}
\newcommand{\kommentar}[1]{}
\newcommand{\Oh}{\ensuremath{\mathcal{O}}}

\newcommand{\proofpara}[1]{\smallskip
	
	\noindent\textit{#1.}}

\DeclareMathOperator{\twwithoutN}{{tw}}

\DeclareMathOperator{\off}{off}
\DeclareMathOperator{\desc}{desc}
\DeclareMathOperator{\anc}{anc}
\DeclareMathOperator{\DP}{DP}

\DeclareMathOperator{\height}{height}
\DeclareMathOperator{\src}{sources}
\newcommand{\Dbar}{{\ensuremath{\overline{D}}}\xspace}
\newcommand{\kbar}{{\ensuremath{\overline{k}}}\xspace}
\newcommand{\w}{{\ensuremath{\omega}}\xspace}

\newcommand\Recc[1]{Recurrence~(\ref{#1})}

\newcommand{\predators}[1]{{\ensuremath{N_{>}(#1)}}\xspace}
\newcommand{\prey}[1]{{\ensuremath{N_{<}(#1)}}\xspace}
\newcommand{\sources}{\sourcespersonal{\Food}}
\newcommand{\sourcespersonal}[1]{{\ensuremath{\src(#1)}}\xspace}

\newcommand{\yes}{{\normalfont\texttt{yes}}\xspace}
\newcommand{\no}{{\normalfont\texttt{no}}\xspace}

\newcommand{\Wh}[1]{{\normalfont W[#1]}\xspace}
\newcommand{\NP}{{\normalfont{NP}}\xspace}

\newcommand{\FPT}{{\normalfont{FPT}}\xspace}

\newcommand{\XP}{{\normalfont{XP}}\xspace}
\newcommand{\NPcoNPpoly}{{\normalfont{NP~$\not\subseteq$~coNP/poly}}\xspace}

\newcommand{\Instance}{{\ensuremath{\mathcal{I}}}\xspace}

\newcommand{\Tree}{{\ensuremath{\mathcal{T}}}\xspace}
\newcommand{\Food}{{\ensuremath{\mathcal{F}}}\xspace}

\newcommand{\tw}{{\ensuremath{\twwithoutN_{\Food}}}\xspace}
\newcommand{\PD}{\PDsub\Tree}
\newcommand{\PDsub}[1]{{\ensuremath{PD_{#1}}}\xspace}

\newcommand{\spannbaum}[1]{\spannbaumsub{\Tree}{#1}}
\newcommand{\spannbaumsub}[2]{\ensuremath{#1\langle #2 \rangle}\xspace}

\newcommand{\problemdef}[3]{
	\begin{quote}
		\normalsize\textsc{#1} \smallskip \\
		\begin{tabularx}{0.9\textwidth}{@{}l@{\hspace{3pt}}X}
			\normalsize\textbf{Input:}    & \normalsize#2 \\
			\normalsize\textbf{Question:} & \normalsize#3
		\end{tabularx}
	\end{quote}
}

\newcommand{\PROB}[1]{{{\normalfont\textsc{#1}}}\xspace}

\newcommand{\MPD}{\PROB{Max-PD}}
\newcommand{\MPDlong}{\PROB{Maximize Phylogenetic Diversity}}

\newcommand{\VC}{\PROB{Vertex Cover}}
\newcommand{\SC}{\PROB{Set Cover}}
\newcommand{\KP}{\PROB{Knapsack}}

\newcommand{\HS}{\PROB{Hitting Set}}
\newcommand{\DS}{\PROB{Dominating Set}}

\newcommand{\rbnb}{\PROB{Red-Blue Non-Blocker}}
\newcommand{\MCNF}{\PROB{MCNF}}
\newcommand{\MCNFlong}{\PROB{Minimum-Cost Network Flow}}

\usetikzlibrary{shadows}

\newcommand{\PDDlong}{\PROB{Optimizing PD with Dependencies}}
\newcommand{\sPDDlong}{\PROB{Optimizing PD in Vertex-Weighted Food-Webs}}
\newcommand{\cksPDDlong}{\PROB{$k$-colored \sPDDlong}}
\newcommand{\PDDplong}{\PROB{Optimizing PD with Pattern-Dependencies}}
\newcommand{\cDPDDlong}{\PROB{$D$-colored \PDDlong}}
\newcommand{\ckPDDlong}{\PROB{$2$-colored \PDDlong}}
\newcommand{\PDD}{\PROB{PDD}}
\newcommand{\sPDD}{\PROB{s-PDD}}
\newcommand{\cksPDD}{\PROB{$k$-c-\sPDD}}
\newcommand{\PDDp}{\PROB{\PDD-pattern}}
\newcommand{\cDPDD}{\PROB{$D$-c-\PDD}}
\newcommand{\ckPDD}{\PROB{$2$-c-\PDD}}

\newcommand{\HStw}{\PROB{\HS with Tree-Profits}}

\DeclareMathOperator{\distclust}{cvd}
\DeclareMathOperator{\distcoclust}{dist-co-clust}

\newcommand{\todos}[2][]{\todo[#1,color=red!25!green!50]{ #2}}
\newcommand{\todosi}[2][]{\todo[inline,color=red!25!green!50]{ #2}}
\newcommand{\todok}[2][]{\todo[#1,color=red!30]{ #2}}

\usepackage{hyphenat}
\title{Maximizing Phylogenetic Diversity under Ecological Constraints: A Parameterized Complexity Study}

\titlerunning{Parameterized Complexity of PD with Ecological Constraints} 
\author{Christian Komusiewicz}{Friedrich-Schiller-Universität Jena, Germany}{c.komusiewicz@uni-jena.de}{https://orcid.org/0000-0003-0829-7032}{}
\author{Jannik Schestag}{Friedrich-Schiller-Universität Jena, Germany}{j.t.schestag@uni-jena.de}{https://orcid.org/0000-0001-7767-2970}{}


\Copyright{Christian Komusiewicz and Jannik Schestag}

\authorrunning{Komusiewicz and Schestag}
\hideLIPIcs

\newif\ifWABIShort
\newif\ifJournal

\WABIShorttrue
\begin{document}

\maketitle

\begin{abstract}
	In the NP-hard \PDDlong{} (PDD) problem, the input consists of a phylogenetic tree~\Tree over a set of taxa~$X$, a food-web that describes the prey-predator relationships in~$X$, and integers~$k$ and~$D$. The task is to find a set~$S$ of~$k$ species that is viable in the food-web such that the subtree of \Tree obtained by retaining only the vertices of~$S$ has total edge weight at least~$D$. Herein, viable means that for every predator taxon of~$S$, the set~$S$ contains at least one prey taxon.
	
	We provide the first systematic analysis of \PDD{} and its special case~\sPDD{} from a parameterized complexity perspective. For solution-size related parameters, we show that \PDD{} is FPT with respect to~$D$ and with respect to~$k$ plus the height of the phylogenetic tree. Moreover, we consider structural parameterizations of the food-web. For example, we show an FPT-algorithm for the parameter that measures the vertex deletion distance to graphs where every connected component is a complete graph. Finally, we show that \sPDD{} admits an FPT-algorithm for the treewidth of the food-web. This disproves a conjecture of Faller et al.~[Annals of Combinatorics, 2011] who conjectured that \sPDD{} is NP-hard even when the food-web is a tree. 
\end{abstract}

\todosi{15 pages excluding references, the front page (authors, affiliation, keywords, abstract, etc.) and a brief appendix (of up to 5 pages)}

\newpage
\setcounter{page}{1}
\section{Introduction}
Human activity has greatly accelerated the rate at which biological species go extinct.
The conservation of biological diversity is thus one of mankind's most urgent tasks.
The inherently limited amount of resources that one may devote to this task, however, necessitates decisions on which conservation strategies to pursue.
To support such decisions, one needs to incorporate quantitative information on the possible impact and the success likelihood\todos{sounds like probabilities} of conservation strategies.
In this context, one task is to compute an optimal conservation strategy in light of this information.

To find a conservation strategy with the best positive impact, one would ideally aim to maximize the functional\todos{what is that?} diversity of the surviving taxa (species).
However, measuring this diversity is hard or impossible in many scenarios~\cite{MPC+18}.
As a result, 
maximizing phylogenetic diversity has become the standard, albeit imperfect, surrogate for maximizing functional diversity~\cite{GCW+15,ITC+07,MPC+18}.
Informally, phylogenetic diversity measures the evolutionary distance of a set of taxa.
In its most simple form, this measurement is based on an edge-weighted phylogenetic tree~$\Tree$ of the whole set of taxa~$X$, and the phylogenetic diversity of a subset of taxa~$S$ is the sum of the weights of the edges of the subtree of~$\Tree$ obtained by retaining only the taxa of~$S$.
Assuming equal protection costs for all taxa, the task is to find a set~$S$ of at most~$k$ taxa that achieves maximal phylogenetic diversity.
This problem, called \MPDlong~\cite{FAITH1992}, can be solved very efficiently by a greedy algorithm~\cite{FAITH1992,PDinSeconds,Pardi2005,steel}.

Computing an optimal conservation strategy becomes much more difficult, however, when the success likelihood\todos{sounds like probabilities} of a strategy is included in the model.
One way to achieve this is to add concrete survival probabilities\todos{Fine for me here} for protected taxa, leading in its most general form to the NP-hard \textsc{Generalized Noah's Ark Problem}~\cite{hartmann,GNAP}.
This problem formulation, however, still has a central drawback: It ignores that the survival of some taxa may also depend on the survival of other taxa.
This aspect was first considered by Moulton et~al.~\cite{moulton} in the \PDDlong~(\PDD)~problem.
Here, the input additionally contains a directed acyclic graph~$D$ with vertex set~$X$ where an arc~$uv$ is present if the existence of taxa~$u$ provides all the necessary foundations for the existence of~$B$.
In other words,~$D$ models ecological dependencies between taxa.
Now, a taxon~$v$ may survive only if (i)~it does not depend on other taxa at all, that is, it has no incoming arcs, or (ii)~at least one taxon~$u$ survives such that~$D$ contains the arc~$uv$.
The most-wide spread interpretation of such ecological dependency networks are food-webs where the arc~$uv$ means that taxon~$v$ feeds on taxon~$u$.\footnote{We remark that previous works~\cite{moulton,faller} consider a reversed interpretation of the arcs. We define the order such that a source of the network also corresponds to a source of the ecosystem.}
A set of taxa~$X$ where every vertex fulfills~(i) or~(ii) is called \emph{viable}.
The task in \PDD{} is to select a \emph{viable} set of~$k$ taxa that achieves a maximal phylogenetic diversity.
In this work, we study \PDD{} from an algorithmic point of view.

Moulton et al.~\cite{moulton} showed that \PDD~can be solved by the greedy algorithm if the objective of maximizing phylogenetic diversity agrees with the viability constraint in a precise technical sense.
Later, \PDD was conjectured to be \NP-hard in~\cite{spillner}.
This conjecture was confirmed by Faller et al.~\cite{faller}, who showed that \PDD is \NP-hard even if the food-web~$D$ is a tree.
Further, Faller~et~al.~\cite{faller} considered \sPDD,  the special case where the phylogenetic tree is restricted to be a star, and showed that \sPDD is \NP-hard even for food-webs which have a bipartite graph as underlying graph.
Finally, polynomial-time algorithms were provided for very restricted special cases, for example for \PDD when the food-web is a \emph{directed} tree~\cite{faller}.

\todos{T: Giulio Dalla Riva C: What do you mean exactly?}

\todok{ Importance of food-webs~\cite{cirtwill}}

\todok{Why mention the following: ``Famous book \cite{elton}''?}

%Food webs describe energy and biomass flows
%through a community (Lindeman, 1942; Wootton, 1997), represent
%ecosystem functions (Memmott et al., 2007; Reiss et al., 2009;
%Thompson et al., 2012), and can offer insights into the community’s
%overall stability (Neutel et al., 2002; Thébault and Fontaine, 2010).
%Cited from \cite{cirtwill}

\subparagraph{Our contribution}
As \PDD is \NP-hard even on very restricted instances~\cite{faller}, we turn to parameterized complexity in order to overcome this intractability.
In particular, we aim to identify problem-specific parameters~$\kappa$ such that~\PDD can be solved in~$f(k)\cdot |I|^{\Oh(1)}$ time\todos{Define here and in Sec. 2.1?} (these are called FPT-algorithms) or to show that such algorithms are unlikely, by showing W[1]-hardness for the parameter~$\kappa$.
Here, we consider the most natural parameters related to the solution, such as the solution size~$k$ and the threshold of diversity~$D$, and parameters that describe the structure of the input food-web~$D$.
We formally consider the decision problem, where we ask for the existence of a viable solution with diversity at least~$D$, but our algorithms actually solve the optimization problem as well.

Our most important results are the following; \Cref{tab:results} gives an overview.
In \Cref{thm:k+height}, we prove that \PDD is \FPT when parameterized with the solution size~$k$ plus the height of the phylogenetic tree~$\Tree$.
This also implies that \PDD is \FPT with respect to~$D$, the diversity threshold.
\ifJournal
However, both problems, \PDD and \sPDD, are unlikely to admit a kernel of polynomial size.
\fi
We also consider the dual parameter~$\overline{D}$, that is, the amount of diversity that is lost from~$\Tree$ and show that \PDD is W[1]-hard with respect to~$\overline{D}$.

We then consider the structure of the food-web.
In particular, we consider the special case that each connected component of the food-web~$D$ is a complete digraph.
This case is structurally equivalent to the case that each connected component of~$D$ is a star with one source vertex.
Thus, this case describes a particularly simple dependency structure, where taxa are either completely independent or have a common source.
We show that \PDD is NP-hard in this special case while~\sPDD has an FPT-algorithm when parameterized by the vertex deletion distance to this special case.
Finally, we show that \sPDD is \FPT with respect to the treewidth of the food-web and therefore can be solved in polynomial time if the food-web is a tree (\Cref{thm:tw}).
Our result  disproves a conjecture of Faller et al.~\cite[Conjecture~4.2]{faller} that \PDD is NP-hard even when the food-web is a tree.

\ifJournal
\Cref{fig:results} gives results for structural parameters and sets these parameters into relation.
\fi
\begin{table}[t]\centering
	
	\resizebox{\columnwidth}{!}{%
          \begin{tabular}{l ll ll}
            \hline
			Parameter & \multicolumn{2}{c}{\sPDD} & \multicolumn{2}{c}{\PDD} \\
			\hline
			Budget $k$ & \FPT & Thm.~\ref{thm:k-stars} & \XP (\FPT is \emph{open}) & Obs.~\ref{obs:k-XP}\\
			Diversity $D$ & \FPT & Thm.~\ref{thm:D} & \FPT & Thm.~\ref{thm:D}\\
			\ifJournal
			& no poly kernel & Thm.~\ref{thm:D-kernel-PDss} & no poly kernel & Thm.~\ref{thm:D-kernel-PDts}\\
			\fi
			Species-loss $\kbar$ & \Wh{1}-hard, \XP & Prop.~\ref{prop:Dbar},~Obs.~\ref{obs:k-XP}~~ & \Wh{1}-hard, \XP & Prop.~\ref{prop:Dbar},~Obs.~\ref{obs:k-XP}\\
			Diversity-loss $\Dbar$ & \Wh{1}-hard, \XP & Prop.~\ref{prop:Dbar},~Obs.~\ref{obs:k-XP} & \Wh{1}-hard, \XP & Prop.~\ref{prop:Dbar},~Obs.~\ref{obs:k-XP}\\
			\hline
			D.t. Cluster & \FPT & Thm.~\ref{thm:dist-cluster-FPT} & \NP-h for 0 & Thm.~\ref{thm:dist-cluster-hardness}\\
			D.t. Co-Cluster & \FPT & Thm.~\ref{thm:co-cluster} & \FPT & Thm.~\ref{thm:co-cluster}\\
			Treewidth & \FPT & Thm.~\ref{thm:tw} & \NP-h for 1 & \cite{faller}\\
			\ifJournal
			Max Leaf \# & \FPT & Thm.~\ref{thm:tw} & \NP-h for 2 & Cor.~\ref{cor:maxleaf}\\
			D.t. Bipartite & \NP-h for 0 & \cite{faller} & \NP-h for 0 & \cite{faller}\\
			Max Degree & \NP-h for 3 & \cite{faller} & \NP-h for 3 & \cite{faller}
                                                                                 \fi
                                                                                 \hline
		\end{tabular}
	}
	
	\caption{An overview over the parameterized complexity results for \PDD and \sPDD.}
	\label{tab:results}
\end{table}

\subparagraph*{Structure of the paper}
In \Cref{sec:prelims}, we formally define \PDDlong, give an overview of previous results and our contribution, and prove some simple initial results.
\ifJournal
In \Cref{sec:k} and \Cref{sec:D}, we consider parameterization by the budget~$k$ and the threshold of diversity~$D$, the two integers of the input.
In \Cref{sec:kbar}, we consider \PDD with respect to the number of taxa that go extinct and the acceptable loss of diversity.
\else
In \Cref{sec:k}, we study \sPDD and \PDD with respect to $k$, the solution size.
In \Cref{sec:D}, we show that \PDD is \FPT with respect to the desired diversity but W[1]-hard for the acceptable loss of diversity.
\fi
In \Cref{sec:structural}, we consider parameterization by structural parameters of the food-web. 
Finally, in \Cref{sec:discussion}, we discuss future research ideas.
\ifWABIShort
  Theorems, lemmas, and observations marked with~$(\star)$ are deferred to the appendix.
\fi

\section{Preliminaries}
\label{sec:prelims}
\subsection{Definitions}
% \paragraph*{Mathematical Definitions}
For a positive integer $a$,
by $[a]$ we denote the set $\{1,2,\dots,a\}$, and
by $[a]_0$ the set $\{0\}\cup [a]$.
We generalize functions $f:A\to B$, where~$B$ is a family of sets, to handle subsets $A'\subseteq A$ of the domain by defining $f(A') := \bigcup_{a\in A'} f(a)$.

\todosi{  neighbors, spanning tree, connected component}

For any graph~$G$, we write~$V(G)$ and~$E(G)$, respectively, to denote the set of vertices and edges of~$G$.
We write $\{u,v\}$ for an undirected edge between $u$ and $v$.
For a directed edge from~$u$ to~$v$, we write~$uv$ or~$(u,v)$ to increase readability.
For a vertex set~$V'\subseteq V(G)$, we let~$G[V']:=(V',\{e\in E(G) \mid \text{both endpoints of~$e$ are in $V(G)$}\})$ denote the subgraph of~$G$ induced by~$V'$.
Moreover, with~$G - V':=G[V\setminus V']$ we denote the graph obtained from~$G$ by removing~$V'$ and its incident edges.

\subparagraph*{Phylogenetic Trees and Phylogenetic Diversity.}
A tree~$T = (V,E)$ is a directed graph in which the \emph{root} is the only vertex with an in-degree of~zero, each other vertex has an in-degree of~one.
The root is denoted with~$\rho$. 
The \emph{leaves} of a tree are the vertices which have an out-degree of~zero. We refer to the non-leaf vertices of a tree as \emph{internal vertices}. A tree is a star if the root is the only internal vertex and all other vertices are leafs. 
For a given set $X$, a \emph{phylogenetic~$X$-tree~$\Tree=(V,E,\w)$} is a tree~$T=(V,E)$ with an \emph{edge-weight} function~$\w: E\to \mathbb{N}_{>0}$ and a bijective labeling of the leaves with elements from~$X$ where all non-leaves in \Tree have out-degree at least two.
We write $\max_\w$ to denote the biggest edge weight in~$\Tree$.
The set~$X$ is a set of~\emph{taxa}~(species).
Because of the bijective labeling, we interchangeably use the words taxon and leaf.
In biological applications, the set $X$ is a set of taxa, the internal vertices of~$\Tree$ correspond to biological ancestors of these taxa and~$\w(e)$ describes the phylogenetic distance between the endpoints of~$e$, as these endpoints correspond to distinct taxa, we may assume this distance is greater than zero.\todos{Probably add this sentence to intro and remove here.}

For a directed edge~$uv \in E$, we say $u$ is the \emph{parent} of $v$ and $v$ is a \emph{child} of $u$.
If there is a directed path from~$u$ to~$v$ in~$\Tree$ (including when $u=v$), we say that $u$ is an \emph{ancestor} of $v$ and $v$ is a \emph{descendant} of $u$.
The sets of ancestors and descendants of $v$ are denoted by $\anc(v)$ and~$\desc(v)$, respectively.
The set of descendants of~$v$ which are in~$X$ are \emph{offspring} $\off(v)$ of a vertex $v$.
For an edge $e = uv \in E$, we denote $\off(e) = \off(v)$.

 For a tree~$T = (V,E)$ and a vertex set~$V'\subseteq V$, the \emph{spanning tree of~$V'$} is denoted by~$\spannbaumsub{T}{ V' }$.
The~\emph{subtree of~$T$ rooted at~$v$} is~$\spannbaumsub{T}{\{v\} \cup \off(v)}$ and denoted by~$T_v$, for some vertex~$v \in V$.
Given a set of taxa $A \subseteq X$,
let $E_{\Tree}(A)$ denote the set of edges in $e \in E$ with $\off(e)\cap A \neq \emptyset$.
The \emph{phylogenetic diversity} $\PD(A)$ of $A$ is defined by 
\begin{equation}
	\label{eqn:PDdef}
	\PD(A) := \sum_{e \in E_{\Tree}(A)} \w(e).
\end{equation}

In other words, the phylogenetic diversity $\PD(A)$ of a set $A$ of taxa is the sum of the weights of edges which have offspring in~$A$.

\subparagraph*{Food Webs.}
For a given set of taxa $X$, a \emph{food-web~$\Food=(X,E)$ on $X$} is a directed acyclic graph.
If $xy$ is an edge of $E$ then $x$ is \emph{prey} of $y$ and $y$ is a \emph{predator} of $x$.
The set of prey and predators of $x$ is denoted with $\prey{x}$ and $\predators{x}$, respectively.
A taxon $x$ without prey is a \emph{source} and $\sources$ denotes the set of sources in the food-web~\Food.

For a given taxon $x\in X$ we define $X_{\le x}$ to be the set of taxa $X$ which can reach $x$ in \Food.
Analogously, $X_{\ge x}$ is the set of taxa which $x$ can reach in \Food.

For a given food-web~\Food and a set~$Z \subseteq X$ of taxa,
a set of taxa~$A\subseteq Z$ is \emph{$Z$-viable} if~$\sourcespersonal{\Food[A]} \subseteq \sourcespersonal{\Food[Z]}$.
A set of taxa~$A\subseteq X$ is \emph{viable} if $A$ is $X$-viable.
In other words, a set~$A \subseteq Z$ is~$Z$-viable or viable if each vertex with in-degree~0 in~$\Food[A]$ also has in-degree~0 in~$\Food[Z]$ or in~\Food, respectively.

\subparagraph*{Problem Definitions and Parameterizations.}
Formally, the main problem we regard in this paper is defined as follows.

\problemdef{\PDDlong (\PDD)}{
	A phylogenetic~$X$-tree~$\Tree$, a foodweb~$\Food$ on $X$, and integers~$k$ and~$D$.
}{
	Is there a viable set~$S\subseteq X$ such that~$|S|\le k$, and~$\PD(S)\ge D$?
}\todos{Probably add to intro.}

Additionally in \sPDDlong~(\sPDD) we consider the special case of \PDD in which the phylogenetic~$X$-tree~$\Tree$ is a star.\todos{Name is given by Faller et al.}

Throughout the paper, we adopt the common convention that $n$ is the number of taxa~$|X|$
and we let $m$ denote the number of edges in the food-web $|E(\Food)|$.
Observe that \Tree has $\Oh(n)$~edges.

For an instance $\Instance = (\Tree,\Food,k,D)$ of \PDD, we define $\Dbar : = \PD(X)-D = \sum_{e \in E}\w(e) - D$.
Informally, $\Dbar$ is the acceptable loss of diversity: If we save a set of taxa $A \subseteq X$ with $\PD(A)\geq D$, then the total amount of diversity we lose from $\Tree$ is at most $\Dbar$.
Similarly, we define~$\kbar := |X| - k$.
That is, \kbar is the minimum number of species that need to become extinct.

\todosi{Also define structural parameters?}

\subparagraph*{Parameterized Complexity.}
Throughout this paper, we consider a number of parameterizations of \PDD and \sPDD. 
For a detailed introduction to parameterized complexity refer to the standard monographs~\cite{cygan,downeybook}; we only give a brief overview here.

A parameterization of a problem $\Pi$ associates with each input instance $\Instance$ of $\Pi$ a specific \emph{parameter}~$\kappa$.
A parameterized problem $\Pi$ is \emph{fixed-parameter tractable} (\FPT) with respect to some parameter~$\kappa$ if there exists an algorithm solving every instance $(\Instance,\kappa)$ of $\Pi$ in time~$f(\kappa)\cdot |\Instance|^{\Oh(1)}$.
A parameterized problem $\Pi$ is \emph{slice-wise polynomial} (\XP) with respect to some parameter~$\kappa$ if there exists an algorithm solving every instance $(\Instance,\kappa)$ of $\Pi$ in time~$|\Instance|^{f(\kappa)}$.
Here, in both cases, $f$ is some computable function only depending on~$\kappa$.
Parameterized problems that are \emph{\Wh{1}-hard} are believed not to be \FPT.
%
%Observe that if a parameterized problem is \FPT, it also is \XP.
%For convenience reasons, the word parameterized is usually omitted.
We use the $\mathcal O^*$-notation which omits factors polynomial in the input size.

\subparagraph*{Color Coding.}
For an in-depth treatment of color coding, we refer the reader to~\cite[Sec.~5.2~and~5.6]{cygan} and~\cite{alon}.
Here, we give some definitions which we use throughout the paper.

\label{def:perfectHashFamily}
For integers~$n$ and~$k$,
an~\emph{$(n,k)$-perfect hash family $\mathcal{H}$} is a family of functions $f: [n] \to [k]$ such that for every subset~$Z$ of~$[n]$ of size~$k$, some $f \in \mathcal{H}$ exists which is injective when restricted to~$Z$.
For any integers $n,k \geq 1$
an~$(n,k)$-perfect hash family which contains~$e^k k^{\Oh(\log k)} \cdot \log n$ functions can be constructed in time $e^k k^{\Oh(\log k)} \cdot n \log n$~\cite{Naor1995SplittersAN,cygan}.

\subsection{Preliminary observations}
We start with some observations and reduction rules which we use throughout the paper.

\begin{observation}
	\label{obs:viable}
	Let $\Food$ be a food-web.
	A set~$A \subseteq X$ is viable if and only if there are edges~$E_A \subseteq E(\Food)$ such that every connected component in the graph~$(A,E_A)$ is a tree with root in $\sources$.
\end{observation}
\begin{proof}
	If~$A$ is viable then $\sourcespersonal{\Food[A]}$ is a subset of \sources.
	It follows that for each taxon~$x \in A$, either~$x$ is a source in~\Food or~$A$ contains a prey~$y$ of~$x$.
	Conversely, if such a graph~$(A,E_A)$ exists then explicitly the sources of~$\Food[A]$ are a subset of~\sources.
\end{proof}

\begin{observation}
	\label{obs:solution-size}
	Let $\Instance = (\Tree,\Food,k,D)$ be a \yes-instance of \PDD.
	Then $k > n$ or a viable set~$S \subseteq X$ with~$\PD(S) \ge D$ exists which has  size exactly~$k$.
\end{observation}
\begin{proof}
	Let~$S$ be a solution for \Instance.
	If~$S$ has a size of~$k$, nothing remains to show.
	Otherwise, observe that $S\cup \{x\}$ is viable and $\PD(S\cup \{x\}) \ge \PD(S)$ for each taxon~$x \in (\predators{S} \cup \sources) \setminus S$.
	Because $(\predators{S} \cup \sources) \setminus S$ is non-empty unless $S = X$,
	we conclude that~$S\cup \{x\}$ is a solution and iteratively, there is a solution of size~$k$. 
\end{proof}

\ifWABIShort
\begin{observation}[$\star$]
\else
\begin{observation}
\fi
	\label{obs:top-predator}
	Let $\Instance = (\Tree,\Food,k,D)$ be an instance of \PDD.
	In $\Oh(|\Instance|^2)$~time one can compute an equivalent instance $\Instance' := (\Tree',\Food',k',D')$ of \PDD with only one source in $\Food'$ and~$D' \in \Oh(D)$.
\end{observation}
\newcommand{\proofTopPredator}{
\ifWABIShort
\begin{proof}[Proof of \Cref{obs:top-predator}]
\else
\begin{proof}
\fi
	\proofpara{Construction}
	Let $\Instance = (\Tree,\Food,k,D)$ be an instance of \PDD.
	Add a new taxon $\star$ to \Food and add edges from $\star$ to each taxon~$x$ of~$\sources$ to receive $\Food'$.
	To receive $\Tree'$, add $\star$ as a child to the root $\rho$ of \Tree and set $\w'(\rho \star) = D+1$ and~$\w'(e) = \w(e)$ for each $e\in E(\Tree)$.
	Finally, set $k' := k+1$ and $D' := 2\cdot D + 1$.
	
	\proofpara{Correctness}
	All steps can be performed in $\Oh(|\Instance|^2)$~time.
	Because $S\subseteq X$ is a solution for~\Instance if and only if $S\cup \{\star\}$ is a solution for $\Instance'$,
	the instance $\Instance' = (\Tree',\Food',k+1,2\cdot D+1)$ is a \yes-instance of \PDD if and only if \Instance is a \yes-instance of~\PDD.
\end{proof}
}
\ifWABIShort
\else
\proofTopPredator
\fi

\begin{rr}
	\label{rr:each-taxon-savable}
	Let~$R \subseteq X$ be the set of taxa which have a distance of at least~$k$ to every source.
	Then, set~$\Food' := \Food - R$ and~$\Tree' := \Tree - R$.
\end{rr}
\begin{lemma}
	\label{lem:each-taxon-savable}
	\Cref{rr:each-taxon-savable} is correct and can be applied exhaustively in $\Oh(n+m)$~time.
\end{lemma}
\begin{proof}
	By definition, each viable set of taxa which has a size of~$k$ is disjoint from~$R$.
	Therefore, the set $R$ is disjoint from every solution.
	The set $R$ can be found in~$\Oh(n+m)$ time by breadth-first search.
	This is also the total running time since one application of the rule is exhaustive.
\end{proof}

\begin{rr}
	\label{rr:maxw<D}
	Apply \Cref{rr:each-taxon-savable} exhaustively.
	If~$\max_\w \ge D$ return \yes.
\end{rr}
After \Cref{rr:each-taxon-savable} has been applied exhaustively,
for any taxon~$x\in X$ there is a viable set~$S_x$ of size at most~$k$ with~$x \in S_x$.
If edge~$e$ has weight at least~$D$, then for each taxon~$x$ which is an offspring of~$e$,
the set~$S_x$ is viable, has size at most~$k$, and~$\PD(S_x) \ge \PD(\{x\}) \ge D$.
So,~$S_x$ is a solution.

\begin{rr}
	\label{rr:redundant-edges}
	Given an instance~$\Instance = (\Tree,\Food,k,D)$ of \PDD in with $vw, uw \in E(\Food)$ for each~$u\in \prey{v}$.
	If~$v$ is not a source, then remove $vw$ from $E(\Food)$.
\end{rr}
\begin{lemma}
	\label{lem:redundant-edges}
	\Cref{rr:redundant-edges} is correct and can be applied exhaustively in $\Oh(n^3)$~time.
\end{lemma}
\begin{proof}
	\proofpara{Correctness}
	If $\Instance'$ is a \yes-instance, then so is $\Instance$.
	
	Conversely, let $\Instance$ be a \yes-instance of \PDD with solution $S$.
	If $v\not\in S$, then $S$ is also a solution for instance $\Instance'$.
	If $v\in S$ then because $S$ is viable in $\Food$ and $\prey{v} = \{u\}$ we conclude $u\in S$.
	Consequently, $S$ is also viable in $\Food - vw$, as $w$ still could be fed by $u$ (if $w\in S$).
	
	\proofpara{Running time}
	For a tuple of taxa~$v$ and~$w$, we can check~$\prey{v} \subseteq \prey{w}$ in~$\Oh(n)$.
	Consequently, an exhaustive application of \Cref{rr:redundant-edges} takes $\Oh(n^3)$~time.
\end{proof}

\section{The solution size $k$}
\label{sec:k}
In this section, we consider parameterization by the size of the solution $k$.
First, we observe that \PDD is \XP when parameterized by $k$ and \kbar.
In \Cref{sec:k-stars} we show that \sPDD is \FPT with respect to $k$.
We generalize this result in \Cref{sec:k+height} by showing that \PDD is \FPT when parameterized by $k+\height_{\Tree}$.
Recall that~$\kbar := n - k$.

\begin{observation}
	\label{obs:k-XP}
	\PDD can be solved in $\Oh(n^{k + 2})$ and $\Oh(n^{\kbar + 2})$~time.
\end{observation}
\begin{proof}
	\proofpara{Algorithm}
	Iterate over the sets $S$ of $X$ of size $k$.
	Return \yes if there is a viable set~$S$ with $\PD(S) \ge D$.
	Return \no if there is no such set.
	
	\proofpara{Correctness and Running time}
	The correctness of the algorithm follows from \Cref{obs:solution-size}.
	Checking whether a set $S$ is viable and has diversity of at least $D$ can be done $\Oh(n^2)$~time.
	The claim follows because there are $\binom{n}{k} = \binom{n}{n-k} = \binom{n}{\kbar}$ subsets of $X$ of size~$k$.
\end{proof}

\subsection{\sPDD with $k$}
\label{sec:k-stars}
We show that \sPDD is \FPT when parameterized by the size of the solution~$k$.

\begin{theorem} \label{thm:k-stars}
	\sPDD can be solved in $\Oh(2^{3.03 k + o(k)} \cdot nm \cdot \log n)$~time.
\end{theorem}

In order to prove this theorem, we color the taxa and require that a solution should contain at most one taxon of each color.
Formally, the auxiliary problem which we consider is defined as follows.
In \cksPDDlong~(\cksPDD), alongside the usual input~$(\Tree,\Food,k,D)$ of \sPDD, we are given a coloring~$c: X\to [k]$ which assigns each taxon a \emph{color}~$c(x) \in [k]$.
We ask whether there is a viable set $S \subseteq X$ of taxa such that $\PD(S) \ge D$, and $c(S)$ is \emph{colorful}.
A set~$c(S)$ is colorful if~$c$ is injective on~$S$.
Observe that each colorful set~$S$ holds $|S| \le k$.
We continue to show how to solve \cksPDD before we apply tools of the color-coding toolbox to extend this result to the uncolored version.

\begin{lemma}
	\label{lem:k-stars}
	\cksPDD can be solved in $\Oh(3^k \cdot n \cdot m)$~time.
\end{lemma}
\begin{proof}
	\proofpara{Table definition}
	Let $\Instance = (\Tree,\Food,k,D,c)$ be an instance of \cksPDD and by \Cref{obs:top-predator} we assume that~$\star \in X$ is the only source in \Food.
	
	Given~$x\in X$, a set of colors $C\subseteq [k]$, and a set of taxa $X'\subseteq X$:
	A set~$S \subseteq X' \subseteq X$ is~\emph{$(C,X')$-feasible} if
	\ifJournal
	\begin{enumerate}[a)]
	\else
	\begin{inparaenum}[a)]
	\fi
		\item\label{it:ka}$c(S) = C$,
		\item\label{it:kb}$c(S)$ is colorful, and
		\item\label{it:kc}$S$ is $X'$-viable.
	\ifJournal
	\end{enumerate}
	\else
	\end{inparaenum}
	\fi
	
	We define a dynamic programming algorithm with tables $\DP$ and $\DP'$.
	For~$x\in X$, $C\subseteq [k]$ we want entry~$\DP[x,C]$ to store the maximum~$\PD(S)$ of~$(C,X_{\ge x})$-feasible sets~$S$.
	Recall $X_{\ge x}$ is the set of taxa which $x$ can reach in $\Food$.
	If no~$(C,X_{\ge x})$-feasible set $S \subseteq X'$ exists, we want $\DP[x,C]$ to store~$-\infty$. (Or a big negative constant.)
	
	In other words, in $\DP[x,C]$ we store the biggest phylogenetic diversity of a set~$S$ which is $X_{\ge x}$-viable and $c$ bijectively maps~$S$ to~$C$.
	
	For any taxon $x$, let $y_1,\dots,y_q$ be an arbitrary but fixed order of $\predators{x}$.
	In the auxiliary table~$\DP'$, we want entry $\DP'[x,p,C]$ for~$p\in [q]$, and $C \subseteq [k]$ to store the maximum~$\PD(S)$ of~$(C,X')$-feasible sets $S \subseteq X'$, where $X' = \{x\} \cup X_{\ge y_1} \cup \dots \cup X_{\ge y_p}$.
	If no~$(C,X')$-feasible set $S \subseteq X'$ exists, we want $\DP'[x,p,C]$ to store~$-\infty$.
	
	\proofpara{Algorithm}
	As a basic case for each~$x\in X$ and~$p\in [|\predators{x}|]$
	let $\DP[x,\emptyset]$ and $\DP[x,p,\emptyset]$ store~0
	and
	let $\DP[x,C]$ store~$-\infty$ if~$C$ is non-empty and $c(x) \not\in C$.
	For each~$x\in X$ with~$\predators{x} = \emptyset$, we store~$\w(\rho x)$ in $\DP[x,\{\hat c(x)\}]$.
	
	Fix a taxon $x \in X$.
	For every~$Z \subseteq C \setminus \{c(x)\}$, we set~$\DP'[x,1,\{c(x)\} \cup Z] := \DP[y_1,Z]$.
	To compute further values, once $\DP'[x,p,Z]$ for each $p\in [q-1]$, and every $Z \subsetneq C$ is computed, for~$Z \subseteq C \setminus \{c(x)\}$ we use the recurrence
	\begin{equation}
		\label{eqn:recurrence-k}
		\DP'[x,p+1,\{c(x)\} \cup Z] :=
		\max_{Z'\subseteq Z}
		\DP'[x,p,\{c(x)\} \cup Z\setminus Z']
		+
		\DP[y_{p+1},Z'].
	\end{equation}
	
	Finally, we set $\DP[x,C] := \DP'[x,q,C]$ for every~$C \subseteq [k]$.
	
	We return \yes if~$\DP[\star,C]$ stores~1 for some $C \subseteq [k]$.
	Otherwise, we return \no.

	\proofpara{Correctness}
	The basic cases are correct.

	The tables are computed first for taxa further away from the source and with an increasing size of $C$.
	Assume that for a fixed taxon $x$ with predators~$y_1, \dots, y_q$ and a fixed $p\in [q]$, the entries~$\DP[x',Z]$ and~$\DP'[x,p',Z]$ for each~$x' \in \predators{x}$, for each $p'\in [p]$, and every $Z \subseteq [k]$ store the desired value.
	Fix a set~$C \subseteq [k]$ with $c(x) \in C$.
	We show that if~$\DP'[x,p+1,C]$ stores~$d$ then there is a~$(C,X')$-feasible set $S \subseteq X' \cup X_{\ge y_{p+1}}$ for $X' := \{x\} \cup X_{\ge y_1} \cup \dots \cup X_{\ge y_{p}}$ with $\PD(S) = d$.
	Afterward, we show that if $S \subseteq X' \cup X_{\ge y_{p+1}}$ with $\PD(S) = d$
	is a~$(C,X')$-feasible set then~$\DP'[x,p+1,C]$ stores at least~$d$. 
	
	If~$\DP'[x,p,C] = d > 0$ then by \Recc{eqn:recurrence-k}, there is a set $Z \subseteq C \setminus \{c(x)\}$ such that $\DP'[x,p,C\setminus Z] = d_x$ and $\DP[y_{p+1},Z] = d_y$ with $d = d_x + d_y$.
	Therefore, there is a~$(C\setminus Z,X')$-feasible set~$S_x \subseteq X'$ with $\PD(S_x) = d_x$ and a~$(Z,X_{\ge y_{p+1}})$-feasible set~$S_y \subseteq X_{\ge y_{p+1}}$ with~$\PD(S_y) = d_y$.
	Define $S := S_x \cup S_y$ and observe that $\PD(S) = d$.
	It remains to show that $S$ is a~$(C,X' \cup X_{\ge y_{p+1}})$-feasible set.
	First, observe that
	because $C\setminus Z$ and $Z$ are disjoint, we conclude that $c(S)$ is colorful.
	Then, $c(S) = c(S_x) \cup c(S_y) = C\setminus Z \cup Z = C$ where the first equation holds because $c(S)$ is colorful.
	The taxa $x$ and $y_{p+1}$ are the only sources in~$\Food[X_{\ge x}]$ and $\Food[X_{\ge y_{p+1}}]$, respectively.
	Therefore, $x$ is in~$S_x$ and~$y_{p+1}$ is in~$S_y$ unless $S_y$ is empty.
	If~$S_y = \emptyset$ then $S = S_x$ and $S$ is $X' \cup X_{\ge y_{p+1}}$-viable because $S$ is $X'$-viable.
	Otherwise, if~$S_y$ is non-empty then because $S_y$ is~$X_{\ge y_{p+1}}$-viable, we conclude $\sourcespersonal{\Food[S_y]} = \{ y_{p+1} \}$.
	As~$x\in S$ and~$y_{p+1} \in \predators{x}$ we conclude $\sourcespersonal{\Food[S]} = \{x\}$ and so $S$ is~$X' \cup X_{\ge y_{p+1}}$-viable.
	Therefore, $S$ is a~$(C,X' \cup X_{\ge y_{p+1}})$-feasible set.
	
	Conversely, let $S \subseteq X' \cup X_{\ge y_{p+1}}$ be a non-empty~$(C,X' \cup X_{\ge y_{p+1}})$-feasible set with $\PD(S) = d$.
	Observe that $X'$ and $X_{\ge y_{p+1}}$ are not necessarily disjoint.
	We define $S_y$ to be the set of taxa of $X_{\ge y_{p+1}}$ which are connected to $y_{p+1}$ in $\Food[X_{\ge y_{p+1}}]$.
	Further, define~$Z := c(S_y)$ and define $S_x := S \setminus S_y$.
	As $c(S)$ is colorful especially $c(S_x)$ and $c(S_y)$ are colorful.
	Thus, $S_y$ is a~$(Z,X_{\ge y_{p+1}})$-feasible time.
	Further, $c(S_x) = C \setminus c(S_y) = C \setminus Z$.
	As $\sourcespersonal{\Food[S]} = \sourcespersonal{\Food[X' \cup X_{\ge y_{p+1}}]} = \{x\}$, we conclude $x\in S$.
	Because \Food is a DAG, $x$ is not in~$X_{\ge y_{p+1}}$ and so~$x$ is in~$S_x$.
	Each vertex of $S$ which can reach $y_{p+1}$ in $\Food[S]$ is in $F_{\ge y_{p+1}}$ and subsequently in~$S_y$.
	Consequently, because $S$ is $X' \cup X_{\ge y_{p+1}}$-viable we conclude $\sourcespersonal{\Food[S_x]} = \{x\}$.
	Thus, $S_x$ is~$(C \setminus Z, X')$-feasible.
	So, $\DP[y_{p+1},Z] = \PD(S_x)$ and $\DP'[x,p,C\setminus Z] = \PD(S_y)$.
	Hence, $\DP'[x,p+1,C]$ stores at least $\PD(S)$.

	\proofpara{Running time}
	The basic cases can be checked in $\Oh(k)$ time.
	As each $c\in [k]$ in \Recc{eqn:recurrence-k} can either be in $Z'$, in $\{c(x)\} \cup Z \setminus Z'$ or in $[k] \setminus (\{c(x)\} \cup Z)$, all entries of the tables can be computed in $\Oh(3^k \cdot n \cdot m)$~time.
\end{proof}

\newcommand{\proofThmKStars}{
	\begin{proof}[Proof of \Cref{thm:k-stars}]
		\proofpara{Reduction}
		Let $\Instance = (\Tree, \Food, k, D)$ be an instance of \PDD.
		We assume that \Food only has one source by~\Cref{obs:top-predator}.
		
		Let $x_1, \dots, x_{n}$ be the taxa.
		Compute an $(n, k)$-perfect hash family $\mathcal{H}$.
		For every $f \in \mathcal{H}$,
		let $c_f$ be a coloring such that
		$c_f(x_j) = f(j)$ for each~$x_j \in X$.
		
		For every $f \in \mathcal{H}$,
		construct an instance $\Instance_{f} = (\Tree, \Food, k, D, c_f)$ of \cksPDD and solve~$\Instance_{f}$ using~\Cref{lem:k-stars}.
		Return \yes if and only if $\Instance_{f}$ is a \yes-instance for some~$f \in \mathcal{H}$.

		\proofpara{Correctness}
		We show that if \Instance has a solution $S$ then there is an~$f \in \mathcal{H}$ such that $\Instance_f$ is a \yes-instance of \cksPDD.
		Let $S$ be a solution for \Instance.
		Thus, $S$ is viable, $\PD(S) \ge D$, and~$S$ has a size of at most~$k$.
		We may assume $|S| = k$ by \Cref{obs:solution-size}.
		By the definition of~$(n, k)$-perfect hash families, a function $f\in \mathcal{H}$ exists
		such that~$c_f(S)$ is colorful.
		So, $S$ is a solution for $\Instance_f$.
		Conversely,
		a solution of $\Instance_f$ for any~$f\in \mathcal{H}$ is a solution for~\Instance.

		\proofpara{Running Time}
		The instances $\Instance_f$ can be constructed in $e^k k^{\Oh(\log k)} \cdot n \log n$ time.
		An instance of \cksPDD can be solved in $\Oh(3^k \cdot n \cdot m)$~time, and the number of instances is~$|\mathcal{C}| = e^k k^{\Oh(\log k)} \cdot \log n$.
		Thus, the total running time is 
		$\Oh^*(e^k k^{\Oh(\log k)} \log n \cdot (3^k\cdot nm))$ which simplifies to $\Oh((3e)^k \cdot 2^{\Oh(\log^2(k))} \cdot nm \cdot \log n)$.
	\end{proof}
}

For the following proof we construct a perfect hash family $\mathcal{H}$ which is defined in \Cref{def:perfectHashFamily}
Central for proving \Cref{thm:k-stars} is to define and solve an instance of \cksPDD for each function in $\mathcal{H}$.
\ifWABIShort
The proof uses standard color-coding and therefore is deferred to the appendix.
\else
\proofThmKStars
\fi

Unluckily, we were unable to prove whether \PDD is \FPT with respect to $k$.\todos{Probably pimp this sentence to not stand there so lonely.}

\subsection{\PDD with $k+\height_{\Tree}$}
\label{sec:k+height}
In this subsection, we generalize the result of the previous subsection and we show that \PDD is \FPT when parameterized with the size of the solution~$k$ plus~$\height_{\Tree}$, the height of the phylogenetic tree.\todok{Note: Height is here unrelated to branch lengths} 
This algorithm uses the techniques of color-coding, data reduction by reduction rules, and the enumeration of trees.

\begin{theorem} \label{thm:k+height}
	\PDD can be solved in~$\Oh^*(K^K \cdot 2^{4.47 K + o(K)})$~time where $K := k\cdot \height_{\Tree}$.
\end{theorem}

We define a \emph{pattern-tree}~$\Tree_P = (V_P,E_P,c_P)$ to be a tree~$(V_P,E_P)$ with a vertex-coloring $c_P: V_P \to [k\cdot \height_{\Tree}]$.
Recall that~$\spannbaum{Y}$ is the spanning tree of the vertices in~$Y$.

To show the result above, we use a subroutine for solving the following problem.
In \PDDplong~(\PDDp), we are given alongside the usual input~$(\Tree,\Food,k,D)$ of \PDD a pattern-tree~$\Tree_P = (V_P,E_P,c_P)$, and a vertex-coloring~$c: V(\Tree) \to [k\cdot \height_{\Tree}]$.
We ask whether there is a viable set $S \subseteq X$ of taxa such that~$S$ has a size of at most~$k$, $c(\spannbaum{ S\cup\{\rho\} })$ is colorful, and $\spannbaum{ S\cup\{\rho\} }$ and $\Tree_P$ are \emph{color-equal}.
That is, there is an edge~$uv$ of $\spannbaum{ S\cup\{\rho\} }$ with~$c(u) = c_u$ and $c(v) = c_v$ if and only if there is an edge $u'v'$ of $\Tree_P$ with~$c(u') = c_u$ and~$c(v') = c_v$.
Informally, given a pattern-tree and want that it matches in colors with the spanning tree induced by the root and a solution.

With these reduction rules, we can reduce the phylogenetic tree of a given instance of \PDDp to only be a star and then solve \PDDp by applying~\Cref{thm:k-stars}.

Next we present reduction rules with which we can reduce the phylogenetic tree in an instance of \PDDp to be a star which subsequently can be solved with~\Cref{thm:k-stars}.
Afterward, we show how to apply this knowledge to compute a solution for \PDD.

\begin{rr}
	\label{rr:edge-original}
	Let~$uv$ be an edge of~$\Tree$.
	If there is no edge $u'v'\in E_P$ with $c_P(u') = c(u)$ and $c_P(v') = c(v)$, then set~$\Tree' := \Tree - \desc(v)$ and~$\Food' := \Food - \off(v)$.
\end{rr}
\ifWABIShort
\else
\begin{lemma}
	\label{lem:edge-original}
	\Cref{rr:edge-original} is correct and can be applied exhaustively in $\Oh(n^3)$~time.
\end{lemma}
\fi
\newcommand{\proofRREdgeOrginal}{
\ifWABIShort
\begin{proof}[Correctness of~\Cref{rr:edge-original}]
\else
\begin{proof}
\fi
	\proofpara{Correctness}
	Assume $S\subseteq X$ is a solution of the instance of \PDDp.
	As there is no edge $u'v'\in E_P$ with $c_P(u') = c(u)$ and $c_P(v') = c(v)$ we conclude that $S\cap \desc(v) = \emptyset$ and so the reduction rule is safe.
	
	\proofpara{Running Time}
	To check whether \Cref{rr:edge-original} can be applied, we need to iterate over both $E(\Tree)$ and $E_P$.
	Therefore, a single application can be executed in $\Oh(n^2)$~time.
	In each application of \Cref{rr:edge-original} we remove at least one vertex so that an exhaustive application can be computed in $\Oh(n^3)$~time.
\end{proof}
}
\ifWABIShort
\else
\proofRREdgeOrginal
\fi

\begin{rr}
	\label{rr:edge-pattern}
	Let~$u'v'$ be an edge of~$\Tree_P$.
	For each vertex $u\in V(\Tree)$ with $c(u) = c_P(u')$ such that~$u$ has no child~$v$ with $c(v) = c_P(v')$,
	set~$\Tree' := \Tree - \desc(v)$ and~$\Food' := \Food - \off(v)$.
\end{rr}
\ifWABIShort
\else
\begin{lemma}
	\label{lem:edge-pattern}
	\Cref{rr:edge-pattern} is correct and can be applied exhaustively in $\Oh(n^3)$~time.
\end{lemma}
\fi
\newcommand{\proofRREdgePattern}{
\ifWABIShort
\begin{proof}[Correctness of~\Cref{rr:edge-pattern}]
\else
\begin{proof}
\fi
	\proofpara{Correctness}
	Let $S$ be a solution for the instance of \PDDp.
	The spanning tree $\spannbaum{ S\cup\{\rho\} }$ contains exactly one vertex $w$ with color $c(u)$.
	As $c(w) = c_P(u')$ we conclude that $w$ has a child $w'$ and $c(w') = c(v')$.
	Consequently, $w \ne u$ and $S\cap \desc(u) = \emptyset$.
	
	\proofpara{Running Time}
	Like in \Cref{rr:edge-original}, iterate over the edges of \Tree and $\Tree_P$.
	Each application either removes at least one vertex or concludes that the reduction rule is applied exhaustively.
\end{proof}
}
\ifWABIShort
\else
\proofRREdgePattern
\fi

\begin{rr}
	\label{rr:food-web}
	If \Tree has leaf set $L$ and \Food has vertex set $X$ with $L \subsetneq X$,
	then compute the set $R \subseteq X$ of vertices $v$ which on every path from $v$ to a source of \Food contain a vertex of~$X\setminus L$,
	set~$\Tree' := \Tree - R$ and~$\Food' := \Food - R$.
\end{rr}
\ifWABIShort
\else
\begin{lemma}
	\label{lem:rr-food-web}
	\Cref{rr:food-web} is correct and can be applied exhaustively in $\Oh(n^2 \cdot m)$~time.
\end{lemma}
\fi
\newcommand{\proofRRFoodWeb}{
\ifWABIShort
\begin{proof}[Correctness of~\Cref{rr:food-web}]
\else
\begin{proof}
\fi
	\proofpara{Correctness}
	By definition, each set $S\subseteq X$ with $S\cap R \ne \emptyset$ is not viable.
	Therefore, \Cref{rr:food-web} is correct.
	
	\proofpara{Running Time}
	Iterate over the taxa~$x\in X\setminus L$ and compute whether there is a path from~$x$ to some vertex of $\sources$ in $\Food[X\setminus L]$.
	If not add $x$ to $R$.
	With breadth-first search, this computation can be done in~$\Oh(m\cdot n^2)$~time, where the quadratic factor arises from the iteration over $X\setminus L$ and $\sources$.
	A single application is exhaustive.\todos{True??}
\end{proof}
}
\ifWABIShort
\else
\proofRREdgePattern
\fi

\begin{rr}
	\label{rr:internal-vertex}
	Apply Reduction Rules~\ref{rr:edge-original},~\ref{rr:edge-pattern}, and~\ref{rr:food-web} exhaustively.
	Let~$\rho$ be the root of~\Tree and let~$\rho_P$ be the root of~$\Tree_P$.
	Let $v'$ be a grand-child of $\rho_P$ and let~$u'$ be the parent of~$v'$.
	\ifJournal
	\begin{enumerate}[1.]
	\else
	\begin{inparaenum}[1.]
	\fi
		\item For each vertex $u$ of $\Tree$ with~$c(u) = c_P(u')$
		add edges $\rho v$ to $\Tree$ for every child $v$ of $u$.
		\item Set the weight of $\rho v$ to be $\w(uv)$ if~$c(v) \ne c_P(v')$
		or~$\w(uv)+\w(\rho u)$ if~$c(v) = c_P(v')$.
		\item Add edges~$\rho_P w'$ to $\Tree_P$ for every child~$w'$ of~$u'$.
		\item Set~$\Tree_P' := \Tree_P - u'$ and~$\Tree' := \Tree - u$.
	\ifJournal
	\end{enumerate}
	\else
	\end{inparaenum}
	\fi
\end{rr}
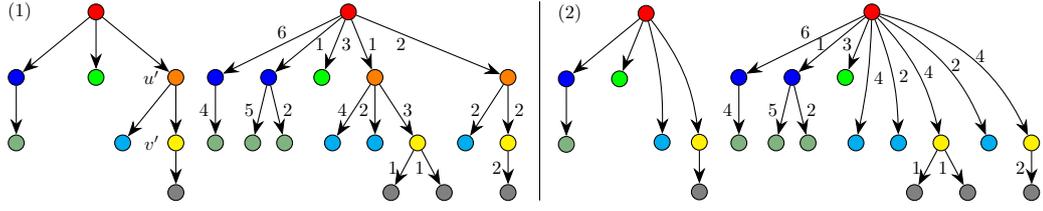
\begin{figure}[t]
\centering
\begin{tikzpicture}[scale=0.8,every node/.style={scale=0.7}]
	\node[draw,fill=red,inner sep=3pt,circle] (root) at (10,10) {};
	
	\node[draw,fill=blue,inner sep=3pt,circle,yshift=5mm,xshift=-15mm] (c1) [below= of root] {};
	\node[draw,fill=green,inner sep=3pt,circle,yshift=5mm,xshift=0mm] (c2) [below= of root] {};
	\node[draw,fill=orange,inner sep=3pt,circle,yshift=5mm,xshift=15mm] (c3) [below= of root] {};
	
	\node[draw,fill=green!40!black!50!,inner sep=3pt,circle,yshift=5mm,xshift=0mm] (c11) [below= of c1] {};
	
	\node[draw,fill=cyan,inner sep=3pt,circle,yshift=5mm,xshift=-10mm] (c31) [below= of c3] {};
	\node[draw,fill=yellow,inner sep=3pt,circle,yshift=5mm,xshift=0mm] (c32) [below= of c3] {};
	
	\node[draw,fill=gray,inner sep=3pt,circle,yshift=8mm,xshift=0mm] (c321) [below= of c32] {};

	\draw[-{Stealth[length=6pt]}] (root) -> (c1);
	\draw[-{Stealth[length=6pt]}] (root) -> (c2);
	\draw[-{Stealth[length=6pt]}] (root) -> (c3);
	
	\draw[-{Stealth[length=6pt]}] (c1) -> (c11);
	\draw[-{Stealth[length=6pt]}] (c3) -> (c31);
	\draw[-{Stealth[length=6pt]}] (c3) -> (c32);
	
	\draw[-{Stealth[length=6pt]}] (c32) -> (c321);
	
	\node[xshift=14mm] [left= of c3] {$u'$};
	\node[xshift=14mm] [left= of c32] {$v'$};
	\node[xshift=5mm] [left= of root] {(1)};
\end{tikzpicture}
\begin{tikzpicture}[scale=0.8,every node/.style={scale=0.7}]
	\node[draw,fill=red,inner sep=3pt,circle] (root) at (10,10) {};
	
	\node[draw,fill=blue,inner sep=3pt,circle,yshift=5mm,xshift=-25mm] (c1) [below= of root] {};
	\node[draw,fill=blue,inner sep=3pt,circle,yshift=5mm,xshift=-15mm] (c4) [below= of root] {};
	\node[draw,fill=green,inner sep=3pt,circle,yshift=5mm,xshift=-5mm] (c2) [below= of root] {};
	\node[draw,fill=orange,inner sep=3pt,circle,yshift=5mm,xshift=5mm] (c3) [below= of root] {};
	\node[draw,fill=orange,inner sep=3pt,circle,yshift=5mm,xshift=30mm] (c5) [below= of root] {};
	
	\node[draw,fill=green!40!black!50!,inner sep=3pt,circle,yshift=5mm,xshift=0mm] (c11) [below= of c1] {};
	
	\node[draw,fill=green!40!black!50!,inner sep=3pt,circle,yshift=5mm,xshift=-3mm] (c41) [below= of c4] {};
	\node[draw,fill=green!40!black!50!,inner sep=3pt,circle,yshift=5mm,xshift=3mm] (c42) [below= of c4] {};
	
	\node[draw,fill=cyan,inner sep=3pt,circle,yshift=5mm,xshift=-8mm] (c30) [below= of c3] {};
	\node[draw,fill=cyan,inner sep=3pt,circle,yshift=5mm,xshift=0mm] (c31) [below= of c3] {};
	\node[draw,fill=yellow,inner sep=3pt,circle,yshift=5mm,xshift=8mm] (c32) [below= of c3] {};
	
	\node[draw,fill=cyan,inner sep=3pt,circle,yshift=5mm,xshift=-8mm] (c51) [below= of c5] {};
	\node[draw,fill=yellow,inner sep=3pt,circle,yshift=5mm,xshift=0mm] (c52) [below= of c5] {};
	
	\node[draw,fill=gray,inner sep=3pt,circle,yshift=8mm,xshift=-5mm] (c321) [below= of c32] {};
	\node[draw,fill=gray,inner sep=3pt,circle,yshift=8mm,xshift=5mm] (c322) [below= of c32] {};
	
	\node[draw,fill=gray,inner sep=3pt,circle,yshift=8mm,xshift=0mm] (c521) [below= of c52] {};
	
	\draw[-{Stealth[length=6pt]}] (root) -> node[above] {6} (c1);
	\draw[-{Stealth[length=6pt]}] (root) -> node[right] {1} (c4);
	\draw[-{Stealth[length=6pt]}] (root) -> node[right] {3} (c2);
	\draw[-{Stealth[length=6pt]}] (root) -> node[right] {1} (c3);
	\draw[-{Stealth[length=6pt]}] (root) -> node[left,xshift=-3mm] {2} (c5);
	
	\draw[-{Stealth[length=6pt]}] (c1) -> node[left] {4} (c11);
	\draw[-{Stealth[length=6pt]}] (c4) -> node[left] {5} (c41);
	\draw[-{Stealth[length=6pt]}] (c4) -> node[right] {2} (c42);
	
	\draw[-{Stealth[length=6pt]}] (c3) -> node[left] {4} (c30);
	\draw[-{Stealth[length=6pt]}] (c3) -> node[left] {2} (c31);
	\draw[-{Stealth[length=6pt]}] (c3) -> node[right] {3} (c32);
	
	\draw[-{Stealth[length=6pt]}] (c32) -> node[left] {1} (c321);
	\draw[-{Stealth[length=6pt]}] (c32) -> node[left] {1} (c322);
	
	\draw[-{Stealth[length=6pt]}] (c5) -> node[left] {2} (c51);
	\draw[-{Stealth[length=6pt]}] (c5) -> node[right] {2} (c52);
	
	\draw[-{Stealth[length=6pt]}] (c52) -> node[left] {2} (c521);
\end{tikzpicture}
\begin{tikzpicture}[scale=0.8,every node/.style={scale=0.7}]
	\node[draw,fill=red,inner sep=3pt,circle] (root) at (10,10) {};
	
	\node[draw,fill=blue,inner sep=3pt,circle,yshift=5mm,xshift=-15mm] (c1) [below= of root] {};
	\node[draw,fill=green,inner sep=3pt,circle,yshift=5mm,xshift=-5mm] (c2) [below= of root] {};
	
	\node[draw,fill=green!40!black!50!,inner sep=3pt,circle,yshift=5mm,xshift=0mm] (c11) [below= of c1] {};
	
	\node[draw,fill=cyan,inner sep=3pt,circle,yshift=-7mm,xshift=3mm] (c31) [below= of root] {};
	\node[draw,fill=yellow,inner sep=3pt,circle,yshift=-7mm,xshift=10mm] (c32) [below= of root] {};
	
	\node[draw,fill=gray,inner sep=3pt,circle,yshift=8mm,xshift=0mm] (c321) [below= of c32] {};
	
	\draw[-{Stealth[length=6pt]}] (root) -> (c1);
	\draw[-{Stealth[length=6pt]}] (root) -> (c2);
	
	\draw[-{Stealth[length=6pt]}] (c1) -> (c11);
	\draw[-{Stealth[length=6pt]}] (root) to[bend left=8] (c31);
	\draw[-{Stealth[length=6pt]}] (root) to[bend left=15] (c32);
	
	\draw[-{Stealth[length=6pt]}] (c32) -> (c321);
	
	\node[xshift=5mm] [left= of root] {(2)};
	\draw (8.25,10.2) -> (8.25,6.9);
\end{tikzpicture}
\begin{tikzpicture}[scale=0.8,every node/.style={scale=0.7}]
	\node[draw,fill=red,inner sep=3pt,circle] (root) at (10,10) {};
	
	\node[draw,fill=blue,inner sep=3pt,circle,yshift=5mm,xshift=-25mm] (c1) [below= of root] {};
	\node[draw,fill=blue,inner sep=3pt,circle,yshift=5mm,xshift=-15mm] (c4) [below= of root] {};
	\node[draw,fill=green,inner sep=3pt,circle,yshift=5mm,xshift=-5mm] (c2) [below= of root] {};
	
	\node[draw,fill=green!40!black!50!,inner sep=3pt,circle,yshift=5mm,xshift=0mm] (c11) [below= of c1] {};
	
	\node[draw,fill=green!40!black!50!,inner sep=3pt,circle,yshift=5mm,xshift=-3mm] (c41) [below= of c4] {};
	\node[draw,fill=green!40!black!50!,inner sep=3pt,circle,yshift=5mm,xshift=3mm] (c42) [below= of c4] {};
	
	\node[draw,fill=cyan,inner sep=3pt,circle,yshift=5mm,xshift=-8mm] (c30) [below= of c3] {};
	\node[draw,fill=cyan,inner sep=3pt,circle,yshift=5mm,xshift=0mm] (c31) [below= of c3] {};
	\node[draw,fill=yellow,inner sep=3pt,circle,yshift=5mm,xshift=8mm] (c32) [below= of c3] {};
	
	\node[draw,fill=cyan,inner sep=3pt,circle,yshift=5mm,xshift=-8mm] (c51) [below= of c5] {};
	\node[draw,fill=yellow,inner sep=3pt,circle,yshift=5mm,xshift=0mm] (c52) [below= of c5] {};
	
	\node[draw,fill=gray,inner sep=3pt,circle,yshift=8mm,xshift=-5mm] (c321) [below= of c32] {};
	\node[draw,fill=gray,inner sep=3pt,circle,yshift=8mm,xshift=5mm] (c322) [below= of c32] {};
	
	\node[draw,fill=gray,inner sep=3pt,circle,yshift=8mm,xshift=0mm] (c521) [below= of c52] {};
	
	\draw[-{Stealth[length=6pt]}] (root) -> node[above] {6} (c1);
	\draw[-{Stealth[length=6pt]}] (root) -> node[left] {1} (c4);
	\draw[-{Stealth[length=6pt]}] (root) -> node[left] {3} (c2);
	
	\draw[-{Stealth[length=6pt]}] (c1) -> node[left] {4} (c11);
	\draw[-{Stealth[length=6pt]}] (c4) -> node[left] {5} (c41);
	\draw[-{Stealth[length=6pt]}] (c4) -> node[right] {2} (c42);
	
	\draw[-{Stealth[length=6pt]}] (c32) -> node[left] {1} (c321);
	\draw[-{Stealth[length=6pt]}] (c32) -> node[left] {1} (c322);
	
	\draw[-{Stealth[length=6pt]}] (root) to[bend left=5] node[right] {4} (c30);
	\draw[-{Stealth[length=6pt]}] (root) to[bend left=10] node[right] {2} (c31);
	\draw[-{Stealth[length=6pt]}] (root) to[bend left=15] node[right] {4} (c32);
	
	\draw[-{Stealth[length=6pt]}] (root) to[bend left=20] node[right] {2} (c51);
	\draw[-{Stealth[length=6pt]}] (root) to[bend left=25] node[right] {4} (c52);
	
	\draw[-{Stealth[length=6pt]}] (c52) -> node[left] {2} (c521);
\end{tikzpicture}
\caption{A example for \Cref{rr:internal-vertex}.
	(1) An instance of \PDDp (2) The instance after an application of \Cref{rr:internal-vertex} to the marked vertices. 
	In both instances, the pattern-tree is on the left and the phylogenetic tree is on the right.
	}
\label{fig:rr-internal-vertex}
\end{figure}%
\Cref{fig:rr-internal-vertex} depicts an application of an application of \Cref{rr:internal-vertex}.
\ifWABIShort
\else
\begin{lemma}
	\label{lem:internal-vertex}
	\Cref{rr:internal-vertex} is correct and can be applied exhaustively in $\Oh(n^3m)$~time.
\end{lemma}
\fi
\newcommand{\proofRRInternalVertex}{
\ifWABIShort
\begin{proof}[Correctness of~\Cref{rr:internal-vertex}]
\else
\begin{proof}
\fi
	\proofpara{Correctness}
	Assume that \Instance is a \yes-instance of \PDDp with solution $S$.
	Because $\spannbaum{ S\cup\{\rho\} }$ and $\Tree_P$ are color-equal also $\spannbaumsub{\Tree'}{ S\cup\{\rho\} }$ and $\Tree_P'$ are color-equal.
	Let~$u^*$ and~$w_1$ be the unique vertices in $\spannbaum{ S\cup\{\rho\} }$ with $c(u^*) = c_P(u')$ and $c(w_1) = c_P(v')$.
	Let~$w_2,\dots,w_\ell$ be the other children of $u^*$.
	As $\PDsub{\Tree'}(S)$ is the sum of the weights of the edges of $\spannbaumsub{\Tree'}{ S\cup\{\rho\} }$ we conclude
	$\PDsub{\Tree'}(S) = \PD(S) - (\w(\rho u^*) + \sum_{i=1}^\ell \w(u^* w_i)) + \sum_{i=1}^\ell \w'(\rho w_i)$.
	Since $\w'(\rho w_1) = \w(\rho u^*) + \w(u^* w_1)$ and $\w'(\rho w_i) = \w(u^* w_i)$ for $i\in [\ell] \setminus \{1\}$,
	we conclude that~$\PDsub{\Tree'}(S) = \PD(S) \ge D$.
	Therefore, $S$ is a solution for~$\Instance'$.
	
	The other direction is shown analogously.
	
	\proofpara{Running Time}
	For a given grand-child~$v'$ of~$\rho_P$, one needs to perform~$\Oh(n)$ color-checks and add~$\Oh(n)$ edges.
	As the reduction rule can be applied at most~$|\Tree_P| \in \Oh(K) \in \Oh(n)$~times,
	an exhaustive application takes~$\Oh(n^2)$~time. (Plus the other reduction rules.)
\end{proof}
}
\ifWABIShort
\else
\proofRRInternalVertex
\fi

\ifWABIShort
\begin{lemma}[$\star$]
	\label{lem:reduction-rules}
	Reduction rules \ref{rr:edge-original}, \ref{rr:edge-pattern}, \ref{rr:food-web}, and \ref{rr:internal-vertex} are correct and can be applied exhaustively in~$\Oh(n^2\cdot (n+m))$~time.
\end{lemma}
\else
With these reduction rules, we can reduce the phylogenetic tree of a given instance of \PDDp to only be a star and then solve \PDDp by applying~\Cref{thm:k-stars}.
\fi

\begin{lemma}
	\label{lem:k+height}
	\PDDp can be solved in~$\Oh(2^{3.03 k + o(k)} \cdot n^2 \cdot (n+m))$~time.
\end{lemma}
\begin{proof}
	\proofpara{Algorithm}
	Let an instance~$\Instance = (\Tree, \Food, k, D, \Tree_P = (V_P,E_P,c_P), c)$ of \PDDp be given.
	If there is a vertex $v\in V_P$ and $c_P(v) \not\in c(V(\Tree))$ then return \no.
	If $c(\rho) \ne c_P(\rho_P)$ where $\rho$ and $\rho_P$ are the roots of \Tree and $\Tree_P$ respectively, return \no.
	
	Apply~\Cref{rr:internal-vertex} exhaustively.
	Then, both $\Tree_P$ and \Tree~are stars.
	Return \yes if and only if
	$(\Tree,\Food,k,D)$ is a \yes-instance of \sPDD.

	\proofpara{Correctness}
	If $\Tree_P$ contains a vertex~$v$ with $c_P(v) \not\in c(V(\Tree))$,
	or if $c(\rho) \ne c_P(\rho_P)$,
	then \Instance is a \no-instance.
	\ifWABIShort
	Then, the correctness follows by \Cref{lem:reduction-rules} and \Cref{lem:k-stars}.
	\else
	Then, the correctness follows by \Cref{lem:internal-vertex} and \Cref{lem:k-stars}.
	\fi

	\proofpara{Running time}
	\Cref{rr:internal-vertex} can be applied exhaustively in $\Oh(n^2 \cdot (n+m))$~time.
	With the application of \Cref{lem:k-stars}, the overall running time is $\Oh(3^{k\cdot \height_{\Tree}} \cdot n \cdot m)$.
\end{proof}

To prove \Cref{thm:k+height} we reduce from \PDD to \PDDp and apply \Cref{lem:k+height}.
For this, we use the fact that there are~$n^{n-2}$ labeled directed trees with~$n$ vertices~\cite{Shor95} which can be enumerated in $\Oh(n^{n-2})$~time~\cite{beyer}.
To solve instance \Instance of \PDD, we will check each of these trees as a pattern tree for a given coloring of the phylogenetic tree.
These colorings will be defined with a perfect hash family as defined in \Cref{def:perfectHashFamily}.
Recall that $K=k\cdot \height_{\Tree}$.
\begin{proof}[Proof of \Cref{thm:k+height}]
	\proofpara{Algorithm}
	Let $\Instance = (\Tree, \Food, k, D)$ be an instance of \PDD.
	Let the vertices of~\Tree be~$v_1, \dots, v_{|V(\Tree)|}$.
	Iterate over $i \in [\min\{K,|V(\Tree)|\}]$.
	Compute a~$(|V(\Tree)|,i)$-perfect hash family~$\mathcal{H}_i$.
	Compute the set~$\mathcal{P}_i$ of labeled directed trees with~$i$ vertices.
	
	For every~$\Tree_P = (V_P, E_P, c_P) \in \mathcal{P}_i$ precede as follows.
	Assume that the labels of $\Tree_P$ are in $[i]$.
	For every $f \in \mathcal{H}_i$,
	let $c_f$ be a coloring such that $c_f(v_j) = f(j)$ for each $v_j \in V(\Tree)$.
	
	For every $f \in \mathcal{H}_i$,
	solve instance $\Instance_{\Tree_P,f} := (\Tree, \Food, k, D, \Tree_P, c_f)$ of \PDDp using~\Cref{lem:k+height}.
	Return \yes if and only if~$\Instance_{\Tree_P,f}$ is a \yes-instance for some~$f \in \mathcal{H}_i$ and some~$\Tree_P \in \mathcal{P}_i$.

	\proofpara{Correctness}
	Any solution of an instance $\Instance_{\Tree_P,f}$ of \PDDp clearly is a solution for $\Instance$.

	Conversely, we show that if $S$ is a solution for \Instance, then there are~$\Tree_P$ and~$f$ such that~$\Instance_{\Tree_P,f}$ is a \yes-instance of \PDDp.
	So let~$S$ be a viable set of taxa with $|S| \le k$ and $\PD(S) \ge D$.
	Let $V^* \subseteq V(\Tree)$ be the set of vertices~$v$ that have offspring in $S$.
	It follows~$|V^*| \le \height_{\Tree} \cdot |S| \le K$.
	Then, there is a hash function~$f \in \mathcal{H}_{V^*}$ mapping $V^*$ bijectively to $[|V^*|]$.
	Consequently, $\mathcal{P}_{|V^*|}$ contains a tree~$\Tree_P$ which is isomorphic to~$\Tree[V^*]$ with labels~$c_f$.
	Hence, $\Instance_{\Tree_P,f}$ is a \yes-instance of \PDDp.

	\proofpara{Running Time}
	For a fixed $i\in [K]$, the set $\mathcal{H}_i$ contains $e^i i^{\Oh(\log i)} \cdot \log n$ hash functions and the set~$\mathcal{P}_i$ contains $\Oh(i^{i-2})$ labeled trees.
	Both sets can be computed in $\Oh(i^{i-2} \cdot n\log n)$~time.
	
	Each instance~$\Instance_{\Tree_P,f}$ of \PDDp is constructed in $\Oh(n)$~time and can be solved in~$\Oh(2^{3.03 k + o(k)} \cdot n^3 \cdot m)$~time.
	Thus, the overall running time is~$\Oh(K \cdot e^{K} K^{K - 2 + \Oh(\log K)} \cdot 2^{3.03 k + o(k)} \cdot n^3 m \log n)$,
	which summarizes to~$\Oh(K^K \cdot 2^{3.03 k + 1.44 K + o(K)} \cdot n^3 m \log n)$.
\end{proof}

\ifJournal
\section{The desired diversity $D$}
\label{sec:D}
In this section, we consider parameterization with the required threshold of diversity~$D$.
\else
\section{Parameterization by desired diversity and accepted diversity loss}
\label{sec:D}
In this section, we first consider parameterization with the diversity threshold~$D$. For this parameter, we present an FPT algorithm for \PDD.
Afterward, we show that \sPDD is intractable with respect to~\Dbar, the acceptable loss of phylogenetic diversity.
\fi
As the edge-weights are integers, we conclude that we can return \yes if $k\ge D$ or if the height of the phylogenetic tree~\Tree is at least $D$, after \Cref{rr:each-taxon-savable} has been applied exhaustively. Otherwise,~$k+\height_{\Tree}\in \Oh(D)$ and thus the FPT algorithm for $k+\height_{\Tree}$ (\Cref{thm:k+height}) directly gives an \FPT algorithm for \PDD in that case.

\ifJournal
In \Cref{sec:FPT-D}, we present a faster \FPT-algorithm for the parameter~$D$.
Afterward we show that it is unlikely that a polynomial kernel for~$D$ exists, even in very restricted cases.

\subsection{FPT-algorithm for $D$}
\label{sec:FPT-D}
By \Cref{thm:k+height}, \PDD is \FPT with respect to the desired diversity~$D$.
\fi
Here, we present another algorithm with a faster running time.
To obtain this algorithm, we subdivide edges of the phylogenetic tree according to their edge weights. We then use color coding on the vertices of the subdivided tree. Let us remark that this technique is closely related to an algorithm of Jones and Schestag~\cite{timePD} for another hard problem related to diversity maximization. % and we generalize the technique presented in \Cref{sec:k-stars}.
\ifWABIShort
\begin{theorem}[$\star$]
\else
\begin{theorem}
\fi
	\label{thm:D}
	\PDD can be solved in $\Oh(2^{3.03(D + k) + o(D)} \cdot nm + n^2)$~time.
\end{theorem}

\newcommand{\proofThmD}{
We define \cDPDDlong (\cDPDD), a colored version of \PDD, as follows.
In addition to the usual input of \PDD, we receive two colorings~$c$ and~$\hat c$ which assign each edge $e\in E(\Tree)$ a subset $c(e)$ of $[D]$, called \emph{a set of colors}, which is of size $\w(e)$ and each taxon $x\in X$ a \emph{color} $\hat c(x) \in [k]$.
We extend the function~$c$ to also assign color sets to taxa $x\in X$ by defining $c(x) := \bigcup_{e\in E'} c(e)$ where $E'$ is the set of edges with $x\in \off(e)$.
In~\cDPDD, we ask whether there is a viable set~$S\subseteq X$ of taxa exists such that~$c(S) = [D]$, and~$\hat c$ is colorful.

\ifJournal
In the following we show how to solve \cDPDD and then we show how to apply standard color-coding techniques to reduce from \PDD to \cDPDD.
\fi

Finding a solution for \cDPDD can be done with techniques similar to~\Cref{lem:k-stars}.
\begin{lemma}
	\label{lem:colored-D}
	\cDPDD can be solved in $\Oh(3^{D+k} \cdot n\cdot m)$~time.
\end{lemma}
\begin{proof}
	\proofpara{Table definition}
	Let $\Instance = (\Tree,\Food,k,D,c,\hat c)$ be an instance of \cDPDD and we assume that~$\star \in X$ is the only source in \Food by \Cref{obs:top-predator}.
	
	Given~$x\in X$, sets of colors $C_1\subseteq [D]$, $C_2\subseteq [k]$, and a set of taxa $X'\subseteq X$:
	A set~$S \subseteq X' \subseteq X$ is $(C_1,C_2,X')$-feasible if
	\ifJournal
	\begin{enumerate}[a)]
	\else
	\begin{inparaenum}[a)]
	\fi
		\item\label{it:Da}$C_1$ is a subset of $c(S)$,
		\item\label{it:Db}$\hat c(S) = C_2$,
		\item\label{it:Dc}$\hat c(S)$ is colorful,
		\item\label{it:Dd}$S$ is $X'$-viable.
	\ifJournal
	\end{enumerate}
	\else
	\end{inparaenum}
	\fi
	
	We define a dynamic programming algorithm with tables $\DP$ and $\DP'$.
	We want that entry $\DP[x,C_1,C_2]$ for~$x\in X$ and $C_1\subseteq [D]$, $C_2\subseteq [k]$ to stores~1 if a~$(C_1,C_2,X_{\ge x})$-feasible set~$S$ exists.
	Otherwise, we want $\DP[x,C_1,C_2]$ to store~0.
	
	For any taxon $x \in X$ let $y_1,\dots,y_q$ be an arbitrary but fixed order of $\predators{x}$.
	In the auxiliary table $\DP'$ we want entry $\DP'[x,p,C_1,C_2]$ for~$p\in [q]$, and $C_1\subseteq [D]$, $C_2\subseteq [k]$ to store~1 if a~$(C_1,C_2,X')$-feasible set $S \subseteq X'$ exists, where $X' := \{x\} \cup X_{\le y_1} \cup \dots \cup X_{\le y_p}$.
	If no~$(C_1,C_2,X')$-feasible set $S \subseteq X'$ exists, we want $\DP'[x,p,C_1,C_2]$ to store~0.\todos{Potentially cut here some space for the Conference version by referring to Thm 3.1.}
	
	\proofpara{Algorithm}
	As a basic case for each~$x\in X$ and each~$p\in [|\predators{x}|]$ let~$\DP[x,\emptyset,\emptyset]$ and~$\DP'[x,p,\emptyset,\emptyset]$ store~1.
	Further, let~$\DP[x,C_1,C_2]$ and~$\DP'[x,p,C_1,C_2]$ store~0 if~$C_1 \not\subseteq c(x)$, or~$\hat c(x) \not\in C_2$, or~$|C_2| > |X_{\ge x}|$ for each~$x\in X$ and every~$C_1\subseteq [D]$, $C_2\subseteq [k]$.
	For each~$x\in X$ with~$\predators{x} = \emptyset$, we store~1 in $\DP[x,C_1,C_2]$ if $C_1=C_2=\emptyset$ or if $C_1 \subseteq c(x)$ and $C_2 = \{\hat c(x)\}$.
	Otherwise, we store~0.
	
	Fix a taxon $x\in X$.
	Assume~$\DP[y,C_1,C_2]$ is computed for each~$y\in \predators{x}$, every~$C_1\subseteq [D]$, and every $C_2\subseteq [k]$.
	For~$C_1 \subseteq [D] \setminus c(x)$ and~$C_2 \subseteq [k] \setminus \{\hat c(x)\}$, we set
	\begin{equation}
		\label{eqn:recurrence-D-pre}
		\DP'[x,1,c(x) \cup C_1,\{\hat c(x)\} \cup C_2] := \DP[y_1,C_1,C_2].
	\end{equation}
	
	Fix an integer $p\in [q]$.
	We assume $\DP'[x,p,C_1,C_2]$ is computed for every~$C_1\subseteq [D]$, and every $C_2\subseteq [k]$.
	Then for $C_1 \subseteq [D] \setminus c(x)$ and~$C_2 \subseteq [k] \setminus \{\hat c(x)\}$ we use the recurrence
	\begin{equation}
		\label{eqn:recurrence-D}
		\begin{array}l
			\DP'[x,p+1,c(x) \cup C_1,\{\hat c(x)\} \cup C_2] :=\\
			\max_{C_1'\subseteq C_1, C_2'\subseteq C_2}
				\DP'[x,p,c(x) \cup C_1\setminus C_1',\{\hat c(x)\} \cup C_2\setminus C_2']
				\cdot \DP[y_{p+1},C_1',C_2'].
		\end{array}
	\end{equation}
	Finally, we set $\DP[x,C_1,C_2] := \DP'[x,q,C_1,C_2]$ for every~$C_1\subseteq [D]$, and every $C_2\subseteq [k]$.
	
	Return \yes if~$\DP[\star,[D],C_2]$ stores~1 for some~$C_2 \subseteq [k]$.
	Otherwise, return \no.
	
	\proofpara{Correctness}
	\ifJournal
	The correctness can be shown analogously to the correctness of \Cref{lem:k-stars}.\todos{Ich habe hier einen ausformulierten Beweis im Kommentar, aber dann dachte ich, dass diese Kurzfassung es auch tut. Der ausformulierte Beweis müsste überarbeitet werden.}
	
	\proofpara{Running time}
	The tables $\DP$ and $\DP'$ contain $\Oh(n\cdot m \cdot 2^{D+k})$ entries.
	Each entry in the basic case can be computed in $\Oh(D^2 \cdot n)$~time.
	In Recurrence~(\ref{eqn:recurrence-D}), each color can occur either in $C_1'$, in $c(x) \cup C_1 \setminus C_1'$ or not in $c(x) \cup C_1$.
	Likewise with $\{\hat c(x)\} \cup Z$ and $Z'$.
	Therefore, all values in Recurrence~(\ref{eqn:recurrence-D}) can be computed in $\Oh(3^{D+k} \cdot n\cdot m)$ time
	which is thus also the overall running time.\todos{We could use subset convolutions to improve the result to $\Oh(2^{D+k} \cdot n\cdot m)$. Should we?}
	\else
	The correctness and the running time can be shown analogously to proof of \Cref{lem:k-stars}.
	\fi
\end{proof}

We can now prove~\Cref{thm:D} by showning how the standard \PDD can be reduced to \cDPDD via color coding.

\begin{proof}[Proof of \Cref{thm:D}]
	\proofpara{Reduction}
	Let $\Instance = (\Tree, \Food, k, D)$ be an instance of \PDD.
	We assume that \Food only has one source by~\Cref{obs:top-predator}
	and that $\max_\w < D$ by~\Cref{rr:maxw<D}.
	
	Let~$e_1, \dots, e_{|E(\Tree)|}$ and~$x_1, \dots, x_{n}$ be the edges and taxa of~\Tree, respectively.
	We define integers~$W_0 := 0$ and $W_j := \sum_{i=1}^{j} \w(e_{i})$ for each $j\in [|E(\Tree)|]$.
	Set $W := W_{|E(\Tree)|}$. \todo{Here one or two informal sentences about the idea would help}
	
	Compute a $(W, D)$-perfect hash family $\mathcal{H}_D$ and a~$(n,k)$-perfect hash family $\mathcal{H}_k$.  
	For every $g \in \mathcal{H}_k$,
	let $c_{g,2}$ be a coloring such that $c_{g,2}(x_j) = f(j)$ for each~$x_j \in X$.
	For every $f \in \mathcal{H}_D$,
	let $c_{f,1}$ be a coloring such that $c_{f,1}(e_j) = \{f(W_{j-1}+1), \dots, f(W_j)\}$ for each $e_j \in E(\Tree)$.
	
	For hash functions $f \in \mathcal{H}_D$ and~$g\in \mathcal{H}_k$,
	construct an instance $\Instance_{f,g} = (\Tree, \Food, k, D, c_{f,1}, c_{g,2})$ of \cDPDD.
	Solve every instance~$\Instance_{f,g}$ using~\Cref{lem:colored-D} and return \yes if and only if $\Instance_{f,g}$ is a \yes-instance for some~$f \in \mathcal{H}_D$, $g\in \mathcal{H}_k$.

	\proofpara{Correctness}
	We first show that if \Instance is a \yes instance then $\Instance_{f,g}$ is a \yes-instance for some~$f \in \mathcal{H}_D$, $g\in \mathcal{H}_k$.
	For any set of edges $E'$ with $\w(E') \ge D$, there is a corresponding subset of~$[W]$ of size at least $D$.
	Since $\mathcal{H}_k$ is a $(W, D)$-perfect hash family,~$c_{f,1}(E') = [D]$, for some $f \in \mathcal{H}_D$.
	Analogously, for each set $X'$ of taxa with $|X'| = k$ there is a hash function~$g\in \mathcal{H}_k$ such that $c_{g,2}(X') = [k]$.
	Thus in particular, if $S\subseteq X$ is a solution of size~$k$ for $\Instance$, then $c_{f,1}(S) = [D]$, for some $f \in \mathcal{H}_D$ and $c_{g,2}(S) = [k]$, for some $g \in \mathcal{H}_k$.
	It follows that one of the constructed instances of \cDPDD is a \yes-instance.
	
	Conversely, a solution for $\Instance_{f,g}$ for some~$f \in \mathcal{H}_D$, $g\in \mathcal{H}_k$ is also a solution for $\Instance$.

	\proofpara{Running Time}
	We require $\Oh(n^2)$~time for~\Cref{rr:maxw<D} and~\Cref{obs:top-predator}.
	We can construct~$\mathcal{H}_D$ and~$\mathcal{H}_k$ in $e^D D^{\Oh(\log D)} \cdot W \log W$ time.
	Solving instances of \cDPDD takes $\Oh(3^{D+k}\cdot nm)$~time each, and the number of instances is $|\mathcal{H}_D| \cdot |\mathcal{H}_k| = e^D D^{\Oh(\log D)} \cdot \log W \cdot e^k k^{\Oh(\log k)} \cdot \log n \in e^{D+k + o(D)} \cdot \log W$.
	
	Thus, the total running time is 
	$\Oh(e^{D+k + o(D)} \cdot \log W \cdot (W + 3^{D+k}\cdot nm))$.
	Because $W = \PD(X) < 2n\cdot D$ this simplifies to $\Oh((3e)^{D + k + o(D)} \cdot nm + n^2)$.
\end{proof}
}
\ifWABIShort
\else
\proofThmD
\fi

\ifJournal
\section{$\kbar$ and $\Dbar$}
\label{sec:kbar}
\subsection{Hardness for the acceptable diversity loss $\Dbar$}
\fi
In some instances, the diversity threshold~$D$ may be very large. Then, however, the acceptable loss of diversity~\Dbar could be relatively small.
Recall $\Dbar$ is defined as $\PD(X) - D$.
Encouraged by this observation, recently, several problems on maximizing phylogenetic diversity have been studied with respect to the acceptable diversity loss~\cite{MAPPD,timePD}.
In this section, we show that, unfortunately, \sPDD is already $\Wh 1$-hard with respect to $\Dbar$ even if edge-weights are at most two.

To show this result, we reduce from \rbnb.
In \rbnb, the input is an undirected bipartite graph $G$ with vertex bipartition $V=V_r \cup V_b$ and an integer $k$.
The question is whether there is a set $S\subseteq V_r$ of size at least $k$ such that the neighborhood of $V_r \setminus S$ is $V_b$.
\rbnb is \Wh 1-hard when parameterized by the size of the solution $k$~\cite{downey}.
\begin{proposition}
	\label{prop:Dbar}
	\sPDD is $\Wh 1$-hard with respect to $\Dbar$, even if $\max_\w=2$.
\end{proposition}
\begin{proof}
	\proofpara{Reduction}
	Let $\Instance := (G=(V=V_r \cup V_b,E),k)$ be an instance of \rbnb.
	%	and let the red vertices be labeled $V_r = \{v_1,\dots,v_{|V_r|}\}$
	%	Let $M$ be a sufficiently large integer, i.e. $M > |V_r|$, and
	%	let $W_1,\dots,W_{|V_r|}$ be vertex-sets of size $M$.
	%	Define $W$ to be the union $W_1 \cup \dots \cup W_{|V_r|}$.
	We construct an instance $\Instance' = (\Tree, \Food, k', D)$ of \sPDD as follows.
	\ifJournal
	For better visualization, \Cref{fig:Dbar} depicts an example of this reduction.\todos{Before creating the TikZ, it seamed to be a good idea in my eyes. Now I would like to erase it.}
	\fi
	Let \Tree be a star with root $\rho\not\in V$ and leaves $V$.
	In \Tree, an edge $e=(\rho,u)$ has weight 1 if $u\in V_r$ and
	otherwise~$\w(e) = 2$, if $u\in V_b$.
	Define a food-web \Food with vertices $V$ and
	for each edge~$\{u,v\}\in E$, and every tuple of vertices $u\in V_b$, $v\in V_r$, add an edge $(u,v)$ to \Food.
	%
	%	Further, for each $i\in [|V_r|]$ and $w\in W_i$, add an edge $(w,v_i)$ to \Food.
	%
	Finally, set $k' := |V| - k$ and~$D := 2\cdot |V_b| + |V_r| - k$, or equivalently $\kbar = \Dbar = k$.
	
	\proofpara{Correctness}
	The reduction can be computed in polynomial time.
	We show that if \Instance is a \yes-instance of \rbnb then~$\Instance'$ is a \yes-instance of \PDD. Afterward, we show the converse.
	
	Assume that $\Instance$ is a \yes-instance of \rbnb.
	Therefore, there is a set~$S\subseteq V_r$ of size at least $k$ such that $N_G(V_r\setminus S) = V_b$.
	(We assume $|S| = k$ as $N_G(V_r\setminus S) = V_b$ still holds if we shrink $S$.)
	We define $S' := V\setminus S$ and show that $S'$ is a solution for $\Instance'$.
	The size of~$S'$ is $|V\setminus S| = |V| - |S| = k'$.
	Further, $\PD(S) = 2\cdot |V_b| + |V_r\setminus S| = 2\cdot |V_b| + |V_r| - k = D$.
	By definition, the vertices in $V_r$ are sources.
	Further, because $S$ is a solution for \Instance, each vertex of~$V_b$ has a neighbor in $V_r\setminus S$.
	So, $S'$ is viable and $\Instance'$ is a \yes-instance of \sPDD.
	
	Conversely, let $S'\subseteq V$ be a solution for instance $\Instance'$ of \sPDD.
	Without loss of generality, $S'$ contains $r$ vertices from $V_r$ and $b$ vertices of $V_b$.
	Consequently, $|V| - k \ge |S'| = b + r$ and $2\cdot |V_b| + |V_r| - k = D \le \PD(S') = 2b + r$.
	We conclude $r \le |V| - k - b$ and so~$2b \ge 2\cdot |V_b| + |V_r| - k - r \ge 2\cdot |V_b| + |V_r| - k - (|V| - k - b) = |V_b| + b$.
	Therefore, $b = |V_b|$ and $V_b \subseteq S'$.
	Further, $r = |V_r| - k$.
	We define $S := V_r \setminus S'$ and conclude $|S| = |V_r| - r = k$.
	Because $S'$ is viable, each vertex in $V_b$ has a neighbor in $S'\setminus V_b$.
	Therefore, $S$ is a solution for the \yes-instance \Instance of \rbnb.
\end{proof}

\section{Structural parameters of the food-web}
\label{sec:structural}
In this section, we study how the structure of the food-web affects the complexity of \sPDD and \PDD.
\ifJournal
\else
In the first subsection, we have an in-depth look into the parameterization with respect to the distance of the food-web to a cluster graph, denoted~$\distclust$.
We show that \PDD is \NP-hard even if the food-web is a cluster graph but \sPDD is \FPT when parameterized by~$\distclust$.
Afterward, we show that \PDD is \FPT with respect to the distance to co-cluster and 
\sPDD is \FPT with respect to the treewith of the food-web, donoted by~$\tw$.
\fi

\subsection{Distance to cluster}
In this subsection, we consider the special case that given an instance of \PDD or \sPDD, we need to remove few vertices from the undirected underlying graph of the food-web \Food to receive a cluster graph. Here, a graph is a cluster graph if every connected component is a clique. 
Because \Food is acyclic, we have the following: Every clique in~\Food has exactly one vertex $v_0 \in V(C)$ such that $v_0 \in \prey{v}$ for each $v\in V(C)\setminus \{v_0\}$.
After applying \Cref{rr:redundant-edges} exhaustively to a cluster graph, each connected component of the food-web is an out-star.

\ifJournal
In this subsection we use the following classification of instances.
Recall that~$\spannbaum{Y}$ is the spanning tree of the vertices in~$Y$.
\begin{definition}
An instance $\Instance = (\Tree, \Food, k, D)$ of \PDD or \sPDD is \emph{source-separating} if $\spannbaum{\{\rho\} \cup \sources}$ and $\spannbaum{\{\rho\} \cup (X\setminus \sources)}$ only have $\rho$ as common vertex.
Here, $\rho$ is the root of \Tree.
\end{definition}
\Cref{fig:dist-cluster-flow} depicts in (1) an example of a source-separating instance.
\fi

In \Cref{thm:dist-cluster-FPT}, we show that \sPDD is \FPT with respect to the distance to cluster.
Herein, for a graph $G = (V,E)$ the \emph{distance to cluster $\distclust(G)$} is the smallest number~$d$ such that there exists a set $Y\subseteq V$ of size at most $d$ such that $G - Y$ is a cluster graph. \todok{Does it also work with parameter distance to out-stars? Does it also work with distance to ``every component has a dominating source''?}
\ifJournal
Afterward, we show that \PDD however is \NP-hard on cluster graphs, even if the instance is source-separating.
\else
Afterward, we show that \PDD however is \NP-hard on cluster graphs.
\fi

\begin{theorem}
	\label{thm:dist-cluster-FPT}
	\sPDD can be solved in $\Oh(6^d \cdot n^2 \cdot m \cdot k^2)$~time, when we are given a set~$Y\subseteq X$ of size $d$ such that $\Food-Y$ is a cluster graph.
\end{theorem}
\ifJournal
To prove this theorem, we first consider a special case in the following auxiliary lemma.
\else
To prove this theorem, we first show how to solve the case where we want to save all taxa in~$Y$.
We then show the theorem by reducing the general case to this special case.  This can be done as follows: We iterate over the~$\Oh(2^d)$ subsets of~$Y$. For each each subset~$Z\subseteq Y$, compute whether there is a solution~$S$ for \Instance with~$S \cap Y  = Z$. 
\fi
\begin{lemma}
	\label{lem:dist-cluster-FPT}
	Given an instance $\Instance = (\Tree,\Food,k,D)$ of \sPDD and a set $Y\subseteq X$ of size~$d$ such that $\Food-Y$ is a cluster graph, we can compute whether there is a viable set $S\cup Y$ with~$|S \cup Y|\le k$ and $\PD(S\cup Y) \ge D$ in~$\Oh(3^d \cdot n \cdot k^2)$~time.
\end{lemma}
\begin{proof}
	We provide a dynamic programming algorithm. 
	Let $C_1,\dots,C_c$ be the connected components of $\Food-Y$ and
	let $x_1^{(i)},\dots,x_{|C_i|}^{(i)}$ be an order of $C_i$ such that $(x_{j_1}^{(i)},x_{j_2}^{(i)}) \in E(\Food)$ for~$j_1<j_2$.
	
	\proofpara{Table definition}
	A set $S \subseteq X\setminus Y$ of taxa is \emph{$(\ell,Z)$-feasible}, if $|S| \le \ell$ and $S\cup Z$ is viable.
	The dynamic programming algorithm has tables $\DP$ and $\DP_i$ for each~$i\in [c]$.
	The entry $\DP[i,\ell,Z]$  for~$i\in [c]$, $j\in [|C_i|]$, $\ell\in [k]_0$, $Z \subseteq Y$, and $b\in \{0,1\}$ stores the largest phylogenetic diversity~$\PD(S)$ of an $(\ell,Z)$-feasible set $S \subseteq C_1 \cup \dots \cup C_{i}$ and $-\infty$ if no such set $S$ exists.

	The table entries~$\DP_i[j,b,\ell,Z]$ additionally have a dimension~$b$ with $b\in \{0,1\}$.
	For~$b=0$, an entry $\DP_i[j,b,\ell,Z]$ with $b\in \{0,1\}$ stores the largest phylogenetic diversity~$\PD(S)$ of an $(\ell,Z)$-feasible set $S \subseteq \{ x_1^{(i)}, \dots, x_{j}^{(i)} \}$.
	For $b=1$, additionally some vertex~$v_{j'}^{(i)}$ with~$j' < j$ needs to be contained in $S$.
	
	\proofpara{Algorithm}
	Iterate over the edges of $\Food$.
	For each edge $uv\in E(\Food)$ with $u,v\in Y$, remove all edges incoming at $v$, including $uv$, from $E(\Food)$. After this removal,~$v$ is a new source.
	
	We initialize the base cases of $\DP_i$ by setting $\DP_i[j,0,0,Z]:=0$ for each $i\in [c]$, each~$j\in [|C_i|]$, and every~$Z \subseteq \sources$.
	Moreover,~$\DP_i[1,b,\ell,Z]:=\w(\rho v_1^{(i)})$ if $\ell \ge 1$ and~$Z \subseteq \predators{v_1^{(i)}} \cup \sources$; and $\DP_i[1,b,\ell,Z]:=-\infty$, otherwise.
	
	To compute further values for $j\in [|C_i|-1]$, $b\in \{0,1\}$, and $\ell \in [k]$ we use the recurrences
	\begin{equation}
		\label{eqn:recurrence-dist-cluster-1}
		\DP_i[j+1, b, \ell, Z] = \max\{ \DP_i[j, b, \ell, Z], \DP_i[j, b', \ell-1, Z \setminus \predators{v_{j+1}^{{(i)}}}] + \w(\rho v_{j+1}^{{(i)}})\},
	\end{equation}
	where~$b' = 0$ if there is an edge from a vertex in~$Y$ to~$x_{j+1}^{(i)}$ and otherwise~$b' = 1$.
	
	Finally, we set $\DP[1,\ell,Z] := \DP_1[|C_1|,0,\ell,Z]$ and compute further values with
	\begin{equation}
		\label{eqn:recurrence-dist-cluster-2}
		\DP[i+1, \ell, Z] = \max_{Z' \subseteq Z, \ell' \in [\ell]_0} \DP[i, \ell', Z'] + \DP_{i+1}[|C_{i+1}|, 0, \ell-\ell', Z\setminus Z'].
	\end{equation}
	
	There is a viable set $S\cup Y$ with $|S \cup Y|\le k$ and $\PD(S\cup Y) \ge D$ if and only if~$\DP[c, k-|Y|, Z] \ge D - \PD(Y)$.
	
	\proofpara{Correctness}
	Assume that $\DP$ stores the intended values.
	Then, if~$\DP[c, k-|Y|, Z] \ge D - \PD(Y)$, there is an $(\ell,Z)$-feasible set $S \subseteq X\setminus Y$.
	First, this implies that  $S\cup Y$ is viable, Moreover, since $S$ has size at most $k-|Y|$, we obtain $|S \cup Y| = k$.
	Finally, because~\Tree is a star, $\PD(S) \ge D - \PD(Y)$ implies $\PD(S \cup Y) \ge D$.
	The converse direction can be shown analogously.
	
	It remains to show that $\DP$ and $\DP_i$ store the right values.
	The basic cases are correct.
	Towards the correctness of \Recc{eqn:recurrence-dist-cluster-1},
	as an induction hypothesis, assume that $\DP_i[j,b,\ell,Z]$ stores the desired value %the biggest phylogenetic diversity of an $(\ell,Z)$-feasible set $S\subseteq \{ x_1^{(i)}, \dots, x_{j}^{(i)} \}$ with $S \ne \emptyset$ if $b=1$
	for a fixed $j\in [|C_i|-1]$, each $i\in [c]$, $b\in \{0,1\}$, $\ell\in [k]_0$ and every~$Z\subseteq Y$.
	Let $\DP_i[j+1,b,\ell,Z]$ store $d$.
	We show that there is an $(\ell,Z)$-feasible set $S \subseteq \{ x_1^{(i)}, \dots, x_{j+1}^{(i)} \}$.
	By \Recc{eqn:recurrence-dist-cluster-1}, 
	$\DP_i[j, b, \ell, Z]$ stores~$d$ or 
	$\DP_i[j, 1, \ell-1, Z \setminus \predators{v_{j+1}^{{(i)}}}]$ stores $d - \w(\rho v_{j+1}^{{(i)}})$.
	If $\DP_i[j, b, \ell, Z]$ stores $d$ then there is an $(\ell,Z)$-feasible set $S \subseteq \{ x_1^{(i)}, \dots, x_{j}^{(i)} \} \subseteq \{ x_1^{(i)}, \dots, x_{j+1}^{(i)} \}$.
	If $\DP_i[j, 1, \ell-1, Z \setminus \predators{v_{j+1}^{{(i)}}}]$ stores $d - \w(\rho v_{j+1}^{{(i)}})$ then there is an 
	$(\ell-1,Z \setminus \predators{v_{j+1}^{{(i)}}})$-feasible set $S \subseteq \{ x_1^{(i)}, \dots, x_{j}^{(i)} \}$ containing $x_1^{(i)}$ or $x_{j'}^{(i)} \in \predators{Y}$.
	Consequently, also $S \cup \{x_{j+1}^{(i)}\}$ is $(\ell,Z)$-feasible.
	
	Now, let $S \subseteq \{ x_1^{(i)}, \dots, x_{j+1}^{(i)} \}$ be an $(\ell,Z)$-feasible set.
	We show that $\DP_i[j+1,b,\ell,Z]$ stores at least $\PD(S)$.
	If $S \subseteq \{ x_1^{(i)}, \dots, x_{j}^{(i)} \}$ then we know from the induction hypothesis that $\DP_i[j,b,\ell,Z]$ stores $\PD(S)$ and then also $\DP_i[j+1,b,\ell,Z]$ stores $\PD(S)$.
	If $x_{j+1}^{(i)} \in S$, then $S$ contains $x_{1}^{(i)}$ or some~$x_{j'}^{(i)} \in \predators{Y}$.
	Define $S' := S \setminus \{x_{j+1}^{(i)}\}$.
	Then, $|S'| = \ell-1$ and $S' \cup (Z \setminus \predators{x_{j+1}^{(i)}})$ is viable because $S$ is $(\ell,Z)$-feasible.
	Consequently, $\DP_i[j,1,\ell-1,Z \setminus \predators{x_{j+1}^{(i)}}] \ge \PD(S') = \PD(S) - \w(\rho x_{j+1}^{(i)})$.
	Therefore, $\DP_i[j+1,b,\ell,Z] \ge \PD(S)$.

	Now, we focus on the correctness of \Recc{eqn:recurrence-dist-cluster-2}.
	Let $\DP[i+1,\ell,Z]$ store $d$.
	We show that there is an $(\ell,Z)$-feasible set $S\subseteq C_1 \cup \dots \cup C_{i+1}$ with $\PD(S) = d$.
	Because $\DP[i+1,\ell,Z]$ stores $d$, by \Recc{eqn:recurrence-dist-cluster-2}, there are $Z'\subseteq Z$ and $\ell' \in [\ell]_0$ such that $\DP[i, \ell', Z'] = d_1$, $\DP_{i+1}[|C_{i+1}|, 0, \ell-\ell', Z\setminus Z'] = d_2$ and $d_1 + d_2 = d$.
	By the induction hypothesis, there is an~$(\ell',Z')$-feasible set $S_1 \subseteq C_1 \cup \dots \cup C_{i}$
	and an $(\ell-\ell',Z\setminus Z')$-feasible set $S_2 \subseteq C_{i+1}$
	such that~$\PD(S_1) = d_1$ and~$\PD(S_2) = d_2$.
	Then, $S := S_1 \cup S_2$ holds $|S| \le |S_1| + |S_2| \le \ell' + (\ell - \ell') = \ell$.
	Further, because $Y$ has no outgoing edges $Z' \subseteq \predators{S_1} \cup \sources$ and~$Z\setminus Z' \subseteq \predators{S_2} \cup \sources$.
	Therefore, $Z \subseteq \predators{S} \cup \sources$ and $S\cup Z$ is viable.
	We conclude that~$S$ is the desired set.
	
	Let there be an $(\ell,Z)$-feasible set $S\subseteq C_1 \cup \dots \cup C_{i+1}$ with $\PD(S) = d$.
	We show that $\DP[i+1,\ell,Z]$ stores at least $d$.
	Define $S_1 := S \cap (C_1 \cup \dots \cup C_{i})$ and~$Z' := \predators{S_1} \cap Z$.
	We conclude that $S_1\cap Z'$ is viable.
	Then, $S_1$ is $(\ell',Z')$-feasible, where~$\ell' := |S_1|$.
	Define~$S_2 := S \cap C_{i+1} = S \setminus S_1$.
	Because $S \cup Z$ is viable and $Z$ does not have outgoing edges, we know that $Z \subseteq \predators{S} \cup \sources$.
	So, $Z \setminus Z' \subseteq \predators{S_2} \cup \sources$ and because $|S_2| = |S| - |S_1| = \ell - \ell'$
	we conclude that $S_2$ is $(\ell - \ell',Z \setminus Z')$-feasible.
	Consequently, $\DP[i,\ell',Z'] \ge \PD(S_1)$ and $\DP_{i+1}[|C_{i+1}|,\ell-\ell',Z\setminus Z'] \ge \PD(S_2)$.
	Hence, $\DP[i+1,\ell,Z] \ge \PD(S_1) + \PD(S_2) = \PD(S)$ because \Tree is a star.
	
	\proofpara{Running time}
	The tables $\DP$ and $\DP_i$ for $i\in [c]$ have $\Oh(2^d \cdot n \cdot k)$ entries in total.
	Whether one of the basic cases applies can be checked in linear time.
	We can compute the set $Z \setminus \predators{x}$ for any given $Z\subseteq Y$ and $x\in X$ in $\Oh(d^2)$ time.
	Therefore, the $\Oh(2^d \cdot n \cdot k)$ times we need to apply \Recc{eqn:recurrence-D-pre} consume $\Oh(2^d d^2 \cdot n \cdot k)$ time in total.
	In \Recc{eqn:recurrence-D}, each $x\in Y$ can be in $Z'$, in $Z\setminus Z'$ or in $Y\setminus Z$ so that we can compute all the table entries of $\DP$ in~$\Oh(3^d \cdot n \cdot k^2)$ which is also the overall running time.\todos{Use subset convolutions?}
\end{proof}

\ifJournal
\begin{proof}[Proof (of \Cref{thm:dist-cluster-FPT})]
	\proofpara{Algorithm}
	We iterate over all subsets $Z$ of $Y$.
	Define $R_0 := Y \setminus Z$.
	We want that $Z$ is the set of taxa that are surviving while the taxa in $R_0$ are going extinct.
	Compute the set $R \subseteq X$ of taxa $r$ for which in $\Food$ each path from $r$ to $s$ contains a taxon of~$R_0$, for each $s \in \sources$.
	Compute $\Tree_Z := \Tree - R$ and $\Food_Z := \Food - R$.
	Continue with the next~$Z$ if~$Z \cap R \ne \emptyset$.
	
	Compute with \Cref{lem:dist-cluster-FPT} whether in $\Instance = (\Tree_Z, \Food_Z, k, D)$ there is a viable set $S' = S\cup Z$ such that $|S'| \le k$ and $\PD(S') \ge D$.
	Return \yes, if such a set exists.
	Otherwise, continue with the next subset~$Z$ of $Y$.
	After iterating over all subsets $Z$ of $Y$, return \no.
	
	\proofpara{Correctness}
	Let $Z$ be a given subset of $Y$.
	Assume that there is a solution $S$ with $Y\cap S = Z$.
	Let $x$ be a vertex such that each path from $x$ to any source $s\in \sources{}$ contains a vertex of $Y \setminus Z$.
	Because $S$ is viable, $x \not \in S$.
	
	The rest of the correctness follows by \Cref{lem:dist-cluster-FPT}.
	
	\proofpara{Running time}
	For each subset $Z$ of $Y$, we can compute $\Tree_Z$ and $\Food_Z$ in $\Oh(m\cdot n)$ time.
	By \Cref{lem:dist-cluster-FPT}, we can compute whether there is a solution $S$ with $Y\cap S = Z$ in $\Oh(3^d \cdot n \cdot k^2)$~time.
	Therefore, the overall running time is $\Oh(6^d \cdot n^2 \cdot m \cdot k^2)$.
\end{proof}
\fi

Next, we show that, in contrast to \sPDD, \PDD is \NP-hard even when the food-web is restricted to be a cluster graph. We obtain this hardness
by a reduction from \VC on cubic graphs.
In cubic graphs, the degree of each vertex is \emph{exactly} three.
In \VC we are given an undirected graph~$G = (V,E)$ and an integer~$k$ and ask whether a set~$C\subseteq V$ of size at most~$k$ exists such that~$u\in C$ or~$v\in C$ for each~$\{u,v\} \in E$.
The set~$C$ is called a \emph{vertex cover}.
\VC remains \NP-hard on cubic graphs~\cite{mohar}.

\begin{theorem}
	\label{thm:dist-cluster-hardness}
	\PDD is \NP-hard
	\ifJournal
	on source-separating instances in which the food-web is a cluster graph
	and each connected component has at most four vertices.
	\else
	even if the food-web is a cluster graph.
	\fi
\end{theorem}
\begin{proof}
	\proofpara{Reduction}
	Let $(G,k)$ be an instance of \VC, where $G = (V,E)$ is cubic.
	We define an instance $\Instance = (\Tree,\Food,k',D)$ of \PDD as follows.
	Let $\Tree$ have a root $\rho$.
	For each vertex $v\in V$, we add a child $v$ of $\rho$.
	For each edge $e=\{u,v\}\in E$, we add a child~$e$ of~$\rho$ and two children~$[u,e]$ and~$[v,e]$ of $e$.
	Let $N$ be a big integer.
	We set the weight of~$\rho e$ to~$N-1$ for each edge~$e$ of $E$.
	All other edges of~\Tree have weight 1.
	Additionally, for each edge~$e=\{u,v\}\in E$ we add edges $(u,[u,e])$ and $(v,[v,e])$ to \Food.
	Finally, we set $k' := |E| + k$ and $D := N\cdot |E| + k$.
	
	\proofpara{Correctness}
	Instance \Instance of \PDD is constructed in polynomial time.
	\ifJournal
	The leaves in~\Tree are~$V$ and~$\{ [u,e], [v,e] \mid e=\{u,v\} \in E \}$.
	\fi
	The sources of~\Food are $V$.
	\ifJournal
	Therefore, \Instance is source-separating.
	\fi
	Let~$e_1$, $e_2$, and~$e_3$ be the edges incident with~$v\in V(G)$.
	Each connected component in \Food contains four vertices, $v$, and~$[v,e_i]$ for~$i\in\{1,2,3\}$.
	\ifJournal
	(One could add edges $([v,e_1],[v,e_2])$, $([v,e_1],[v,e_3])$, and~$([v,e_2],[v,e_3])$ to meet the formal requirement of a cluster graph.)
	\fi
	
	We show that $(G,k)$ is a \yes-instance of \VC if and only if \Instance is a \yes-instance of \PDD.
	Let $C\subseteq V$ be a vertex cover of $G$ of size at most $k$.
	If necessary, add vertices to $C$ until $|C|=k$.
	For each edge $e\in E$, let~$v_e$ be an endpoint of $e$ that is contained in $C$. Note that~$v_e$ exists since~$C$ is a vertex cover.
	We show that $S := C \cup \{ [v_e,e] \mid e\in E \}$ is a solution for \Instance:
	The size of $S$ is $|C| + |E| = k + |E|$.
	By definition, for each taxon $[v_e,e]$ we have $v_e\in C\subseteq S$, so $S$ is viable.
	Further, as~$S$ contains the taxon~$[v_e,e]$ for each edge~$e\in E$, we conclude that $\PD(S) \ge N\cdot |E| + \PD(C) = N\cdot |E| + k = D$.
	Therefore, $S$ is a solution of~\Instance.
	
	Let $S$ be a solution of instance \Instance of \PDD.
	Define $C := S \cap V(G)$ and define $S' := S \setminus C$.
	Because $\PD(S) > D$, we conclude that for each $e\in E$ at least one vertex~$[u,e]$ with $u\in e$ is contained in $S'$.
	Thus, $|S'| \ge |E|$ and we conclude that $|C| \le k$.
	Because $S$ is viable we conclude that $u\in C$ for each $[u,e] \in S'$.
	Therefore, $C$ is a vertex cover of size at most~$k$ of~$G$.
\end{proof}

\subsection{Distance to co-cluster}
In this section, we show that \PDD is \FPT with respect to the distance to co-cluster of the food-web.
A  graph is a co-cluster graph if its complement graph is a cluster graph.
Herein, the complement graph is the graph obtained by replacing edges with non-edges and vice versa. In other words, a graph is a co-cluster graph if its vertex set can be partitioned into independent sets such that each pair of vertices from different independent sets is adjacent.

We define an auxiliary problem \HStw in which we are given a universe $\mathcal{U}$, a family of sets $\mathcal{W}$ over $\mathcal{U}$, a $\mathcal{U}$-tree \Tree, and integers $k$ and $D$.
We ask whether there is a set $S \subseteq \mathcal{U}$ of size at most $k$ such that $\PD(S) \ge D$ and $S\cap W \ne \emptyset$ for each $W\in \mathcal{W}$.
\ifWABIShort
Solutions to this problem can be found with a dynamic programming algorithm over the tree, similar to the idea in~\cite{pardi07}.
The proof is therefore in the appendix.
\fi

\ifWABIShort
\begin{lemma}[$\star$]
\else
\begin{lemma}
\fi
	\label{lem:HStw}
	\HStw can be solved in $\Oh(3^{|\mathcal{W}|} \cdot n)$ time.
\end{lemma}
\newcommand{\proofHStw}{
\ifWABIShort
\begin{proof}[Proof of~\Cref{lem:HStw}]
\else
\begin{proof}
\fi
	We adopt the dynamic programming algorithm with which one can solve \MPD with weighted costs for saving taxa in pseudo-polynomial time~\cite{pardi07}.
	
	\proofpara{Table definition}
	We define two tables $\DP$ and and an auxiliary table $\DP'$.
	We want that entry $\DP[v,0,\mathcal{M}]$ for~$v\in V(\Tree)$, and~$\mathcal{M} \subseteq \mathcal{W}$ stores 0 if $\mathcal{M}$ is empty and $\DP[v,1,\mathcal{M}]$ stores the biggest phylogenetic diversity $\PDsub{\Tree_v}(S)$ of a set $S\subseteq \off(v)$ in the subtree $\Tree_v$ rooted at $v$ where~$S\cap M \ne \emptyset$ for each $M\in \mathcal{M}$.
	Otherwise, we store $-\infty$.
	For a vertex~$v \in V(\Tree)$ with children~$w_1,\dots,w_j$,
	in $\DP'[v,i,b,\mathcal{M}]$ we only consider~$S \subseteq \off(w_1) \cup \dots \cup \off(w_i)$.
	
	\proofpara{Algorithm}
	For a leaf $u \in \mathcal{U}$ of \Tree, in $\DP[u,1,\mathcal{M}]$ store 0 if~$u\in M$ for each $M\in\mathcal{M}$.
	Otherwise, store $-\infty$.
	For a given vertex~$v$ of~\Tree, in $\DP[v,0,\mathcal{M}]$ and $\DP'[v,i,0,\mathcal{M}]$ store 0 if $\mathcal{M} = \emptyset$.
	Otherwise, store $-\infty$.
	
	Let $v$ be an internal vertex of \Tree with children $w_1,\dots,w_z$.
	We set $\DP'[v,1,b,\mathcal{M}] := \DP[w_1,b,\mathcal{M}] + b\cdot \w(vw_1)$.
	To compute further values, we use the recurrence
	\begin{equation}
		\label{eqn:recurrence-HStw}
		\DP'[v,i+1,1,\mathcal{M}] =
			\max_{\mathcal{M}', b_1, b_2}
			\DP'[v,i,b_1,\mathcal{M}'] + \DP[w_{i+1},b_2,\mathcal{M}\setminus \mathcal{M}'] + b_2 \cdot \w(vw_{i+1}).
	\end{equation}
	Here, the maximum is taken over $\mathcal{M}' \subseteq \mathcal{M}$ and $b_1,b_2 \in \{0,1\}$ where we additionally require~$b_1+b_2 \ge 1$.	
	Finally, $\DP[v,b,\mathcal{M}] := \DP'[v,z,b,\mathcal{M}]$.
	
	If $\DP[\rho, 1, \mathcal{W}] \ge D$, return \yes.
	Otherwise, return \no. 
	
	\proofpara{Correctness}
	The basic cases are correct.
	Overall, the correctness of Recurrence~(\ref{eqn:recurrence-HStw}) can be shown analogously to~\cite{pardi07}.
	
	\proofpara{Running time}
	As a tree has at most $2n$ vertices, the overall size of $\DP$ and $\DP'$ is $\Oh(2^{|\mathcal{W}|} \cdot n)$.
	In Recurrence~(\ref{eqn:recurrence-HStw}), there are three options for each $M\in \mathcal{W}$.
	Hence that the total number of terms considered in the entire table can be computed in $\Oh(3^{|\mathcal{W}|} \cdot n)$ time.
\end{proof}
}
\ifWABIShort
\else
\proofHStw
\fi

In the following we reduce from \PDD to \HStw.
Herein, we select a subset of the modulator~$Y$ to survive.
Additionally, we select the first taxon~$x_i$ which survives in~$X \setminus Y$.
Because~$\Food - Y$ is a co-cluster graph,~$x_i$ is in a specific independent set~$I \subseteq X$ and any taxon~$X \setminus (I \cup Y)$ feed on~$x_i$.
Then, by selecting taxon~$x_j \in X \setminus (I \cup Y)$, any other taxon in~$X \setminus Y$ has some prey.
Subsequently, a solution is found by \Cref{lem:HStw}.

\begin{theorem}
	\label{thm:co-cluster}
	\PDD can be solved in $\Oh(6^d \cdot n^3)$ time, when we are given a set~$Y\subseteq X$ of size $d$ such that $\Food-Y$ is a co-cluster graph.
\end{theorem}
\begin{proof}
	\proofpara{Algorithm}
	Given an instance $\Instance = (\Tree,\Food,k,D)$ of \PDD.
	Let $x_1,\dots,x_n$ be a topological ordering of $X$ which is induced by \Food.
	Iterate over the subsets $Z$ of $Y$.
	Let $P_Z$ be the sources of~\Food in~$X\setminus Y$ and let $Q_Z$ be $\predators{Z} \setminus Y$, the taxa in $X\setminus Y$ which are being fed by~$Z$.
	Further, define~$R_Z := P_Z \cup Q_Z \subseteq X\setminus Y$.
	Iterate over the vertices $x_i \in R_Z$.
	Let $x_i$ be from the independent set $I$ of the co-cluster graph $\Food-Y$.
	Iterate over the vertices~$x_j \in X \setminus (Y \cup I)$.
	
	For each set~$Z$, and taxa~$x_i$,~$x_j$, with \Cref{lem:HStw} we compute the optimal solution for the case that~$Z$ is the set of taxa of $Y$ that survive while all taxa of $Y\setminus Z$ go extinct, $x_i$ is the first taxon in $X\setminus Y$, and $x_j$ the first taxon in $X \setminus (Y \cup I)$ to survive.
	(The special cases that only taxa from $I \cup Y$ or only from $Y$ survive are omitted here.)

	We define an instance $\Instance_{Z,i,j}$ of \HStw as follows.
	Let the universe $\mathcal{U}_{i,j}$ be the union of $\{ x_{i+1}, \dots, x_{j-1} \} \cap I$ and $\{ x_{j+1}, \dots, x_n \} \setminus Y$.
	\ifJournal
	In other words, we let the taxa in $Y$ and in $\{x_1, \dots, x_{i}\}$ and in $\{x_1, \dots, x_{j}\} \setminus I$ go extinct.
	\fi
	For each taxon~$x\in Z$ compute $\prey{x}$.
	If~$x\not\in \sources$ and~$\prey{x} \cap (Z \cup \{x_i,x_j\}) = \emptyset$, then add $\prey{x} \setminus Y$ to the family of sets $\mathcal{W}_{Z,i,j}$.
	Contract edges~$e\in E(\Tree)$ with~$\off(e) \cap (Z \cup \{x_i,x_j\}) \ne \emptyset$ to receive~$\Tree_{Z,i,j}$.
	Finally, we define $k' := k - |Z| - 2$ and $D' := D - \PD(Z \cup \{x_i,x_j\})$.
	
	Solve~$\Instance_{Z,i,j}$.
	If $\Instance_{Z,i,j}$ is a \yes-instance then return \yes.
	Otherwise, continue with the iteration.
	If $\Instance_{Z,i,j}$ is a \no-instance for every $Z\subseteq Y$, and each $i,j\in [n]$, then return \no.
	
	\proofpara{Correctness}
	We show that if the algorithm returns \yes, then $\Instance$ is a \yes-instance of~\PDD.
	Afterward, we show the converse.
	
	Let the algorithm return \yes.
	Then, there is a set $Z \subseteq Y$, and there are taxa~$x_i \in X \setminus Y$ and~$x_j \in X \setminus (Y \cup V(I))$ $\Instance_{Z,i,j}$ is a \yes-instance of \HStw.
	Here, $I$ is the independent set such that~$v\in V(I)$
	Consequently, there is a set $S \subseteq \mathcal{U}_{i,j}$
	of size at most $k - |Z| - 2$
	such that $\PDsub{\Tree_{Z,i,j}}(S) \ge D' = D - \PD(Z \cup \{x_i,x_j\})$
	and $S \cap W \ne \emptyset$ for each $W\in \mathcal{W}_{Z,i,j}$.
	We show that $S^* := S \cup Z \cup \{x_i,x_j\}$ is a solution for instance \Instance of \PDD.
	Clearly, $|S^*| = |S| + |Z| + 2 \le k$
	and $\PD(S^*) = \PDsub{\Tree_{Z,i,j}}(S) + \PD(Z \cup \{x_i,x_j\}) \ge D$ as $\Tree_{Z,i,j}$ is the $Z \cup \{x_i,x_j\}$-contraction of $\Tree$.
	Further, by definition $x_i \in (\sources \cup \predators{Z}) \setminus Y$.
	Because $\Food-Y$ is a co-cluster graph and $x_j$ is not in $I$, the independent set in which $x_i$ is, we conclude that $x_j\in \predators{x_i}$.
	As $S \cap W \ne \emptyset$ for each $W\in \mathcal{W}_{Z,i,j}$,
	each taxon $x\in Z$ holds $\prey{x} \cap (Z \cup \{x_i,x_j\}) = \emptyset$ or $\prey{x} \cap S \ne \emptyset$ so that $\prey{x} \cap S^* \ne \emptyset$.
	Therefore, $S^*$ is viable and indeed a solution for \Instance.

	Assume now that $S$ is a solution for instance \Instance of \PDD.
	We define $Z := S \cap Y$ and let $x_i$ and $x_j$ be the taxa in $S\setminus Y$, respectively $S\setminus (Y \cup I)$, with the smallest index.
	As before, $I$ is the independent set of $x_i$.
	We show that instance $\Instance_{Z,i,j}$ of \HStw has solution $S^* := S \setminus (Z \cup \{x_i,x_j\})$.
	Clearly, $|S^*| = |S| - |Z| - 2 \le k'$
	and by the definition of $\Tree_{Z,i,j}$ we also conclude $\PDsub{\Tree_{Z,i,j}}(S^*) \ge D'$.
	Let $M \in \mathcal{W}_{Z,i,j}$.
	By definition, there is a taxon $z\in Z$ with $M = \prey{z} \setminus Y$,
	and $z \not\in \sources$,
	and $\prey{z} \cap (Z \cup \{x_i,x_j\}) = \emptyset$.
	Consequently, as $S$ is viable, there is a taxon $x\in S \cap \prey{z}$ so that $S\cap M \ne \emptyset$.
	Hence,~$S^*$ is a solution of instance $\Instance_{Z,i,j}$ of \HStw.

	\proofpara{Running time}
	For a given $Z\subseteq Y$, we can compute
	the topological order $x_1,\dots,x_n$ and the set $R_Y$ in $\Oh(n^2)$ time.
	The iterations over $x_i$ and $x_j$ take $\Oh(n^2)$ time.
	Observe, $|\mathcal{W}_{Z,i,j}| \le |Z|$.
	By \Cref{lem:HStw} checking whether $\Instance_{Z,i,j}$ is a \yes-instance takes $\Oh(3^d n)$~time each.
	The overall running time is $\Oh(6^d \cdot n^3)$ time.
\end{proof}

\subsection{Treewidth}
Conjecture~4.2. formulated by Faller et~al.~\cite{faller} supposes that \sPDD remains \NP-hard on instances where the underlying graph of the food-web is a tree.
In this subsection, we disprove this conjecture by showing that \sPDD can be solved in polynomial time on food-webs which are trees (assuming P$\ne$NP).
We even show a stronger result: \sPDD is \FPT with respect to the treewidth of the food-web.

\ifWABIShort
\begin{theorem}[$\star$]
\else
\begin{theorem}
\fi
	\label{thm:tw}
	\sPDD can be solved in $\Oh(9^\tw \cdot nk)$~time.
\end{theorem}

To show~\Cref{thm:tw}, we define a dynamic programming algorithm over a tree-de\-com\-po\-si\-tion of~\Food.
In each bag, we divide the taxa into three sets indicating that they
	a)~are supposed to go extinct,
	b)~will be saved but still need prey,
	c)~or will be saved without restrictions.
The algorithm is similar to the standard treewidth algorithm for \mbox{\DS~\cite{cygan}}.

\newcommand{\proofTW}{
\ifWABIShort
\begin{proof}[Proof of~\Cref{thm:tw}]
\else
\begin{proof}
\fi
	Let $\Instance = (\Tree, \Food, k, D)$ be an instance of \sPDD.
	We define a dynamic programming algorithm over a nice tree-decomposition~$T$ of~$\Food=(V_\Food,E_\Food)$.% in which vertices are introduced without incident edges and all edges are induced exactly once.
	
	We do not define tree-decompositions.
	Common definitions can be found in~\cite{treedecom,cygan}.
	
	For a node~$t\in T$, let $Q_t$ be the bag associated with~$t$ and let $V_t$ be the union of bags in the subtree of~$T$ rooted a~$t$.
	
	\proofpara{Definition of the Table}
	We index solutions by a partition $R\cup G\cup B$ of $Q_t$, and a non-negative integer $s$.
	For a set of taxa~$Y \subseteq V_t$, we call a vertex $u \in V_t$ \emph{red with respect to $Y$}
	if $u$ is in $Y$ and $u$ has a predator but no prey in $Y$.
	We call $u$ \emph{green with respect to $Y$}
	if $u$ is in $Y$ and
		a)~$u$ is a source in \Food, or
		b)~has prey in $Y$.
	Finally, we call $u$ \emph{black with respect to $Y$}
	if $u$ is not in $Y$.
	
	For a node~$t$ of the tree-decomposition, a partition $R\cup G\cup B$ of $Q_t$, and an integer $s$, a set of taxa~$Y \subseteq V_t$ is called \emph{$(t,R,G,B,s)$-feasible}, if all the following conditions hold.
	\ifJournal
	\begin{enumerate}
	\else
	\begin{inparaenum}
	\fi
		\item[(T1)]\label{it:safeness}Each vertex in $Y\setminus Q_t$ is green with respect to $Y$.
		\item[(T2)]\label{it:red}The vertices $R\subseteq Q_t$ are red with respect to $Y$.
		%are not leaves in \Net and have an incoming but no outgoing edge in $F$.
		\item[(T3)]\label{it:green}The vertices $G\subseteq Q_t$ are green with respect to $Y$. 
		%leaves in \Net or have an outgoing edge in $F$.
		\item[(T4)]\label{it:black}The vertices $B\subseteq Q_t$ are black with respect to $Y$.
		%not leaves in \Net and are not incident with an edge of $F$.
		\item[(T5)]\label{it:taxa}The size of $Y$ is $s$.
	\ifJournal
	\end{enumerate}
	\else
	\end{inparaenum}
	\fi

	In table entry~$\DP[t,A,R,G,B,s]$, we store the largest diversity~$\PD(Y)$ of a $(t,R,G,B,s)$-feasible set~$Y$.
	If there is no $(t,R,G,B,s)$-feasible sets~$Y$, we store $-\infty$.
	Let $r$ be the root of the nice tree-decomposition $T$.
	Then, $\DP[r,\emptyset,\emptyset,\emptyset,k]$ stores an optimal diversity.
	So, we return~\yes if $\DP[r,\emptyset,\emptyset,\emptyset,k]\ge D$ and \no otherwise.
	
	\ifJournal
	In the following, any time a table non-defined entry $\DP[t,R,G,B,s]$ is called (in particular, if $s < 0$), we take $\DP[t,R,G,B,s]$ to be $-\infty$.
	\fi
	
	We regard the different types of nodes of a tree-decomposition separately and discuss their correctness.

	\proofpara{Leaf Node}
	For a leaf~$t$ of~$T$ the bags~$Q_t$ and $V_t$ are empty. We store
	\begin{eqnarray} \label{tw:leaf}
		\DP[t,\emptyset,\emptyset,\emptyset,0] &=& 0.
	\end{eqnarray}
	For all other values, we store $\DP[t,R,G,B,s] = -\infty$.
	
	\Recc{tw:leaf} is correct by definition.

	\proofpara{Introduce Node}
	Suppose that~$t$ is an \emph{introduce node}, that is,~$t$ has a single child~$t'$ with~$Q_t = Q_{t'} \cup \{v\}$.
	We store $\DP[t,R,G,B,s] = \DP[t',R,G,B\setminus \{v\},s]$ if $v\in B$. If $v\in G$ but $v$ is not a source and $\prey{v} \cap (R\cup G) = \emptyset$, we store $\DP[t,R,G,B,s] = -\infty$. Otherwise, we store
	\begin{eqnarray} \label{tw:insertvertex}
		\DP[t,R,G,B,s] &=& \max_{A\subseteq \predators{v} \cap (R\cup G)} \DP[t',R',G',B,s - 1] + \w(\rho v).
	\end{eqnarray}
	Herein, $R'$ and $G'$ are defined as $R' = (R \setminus (\{v\} \cup \predators{v})) \cup A$ and $G' = (G \cup (\predators{v} \cap (R \cup G))) \setminus (A \cup \{v\})$.

	If $v\in B$, then $v\not\in Y$ for any $(t,R,G,B,s)$-feasible set~$Y$.
	If $v \in G$ but is not a source and has no prey in $R\cup G$, then $v$ is not green with respect to~$Y$ for any~$Y$.
	So, these two cases store the desired value.
	Towards the correctness of \Recc{tw:insertvertex}:
	If $v\in G\cup R$, then we want to select $v$ and therefore we need to add $\w(\rho v)$ to the value of~$\DP[t,R,G,B,s]$.
	Further, by the selection of $v$ we know that the predators of $v$ could be green (but maybe are still stored as red) in $t$, but in $t'$ could be red and therefore we reallocate $\predators{v} \cap (R\cup G)$ and let $A$ be red beforehand.

	\proofpara{Forget Node}
	Suppose that~$t$ is a \emph{forget node}, that is,~$t$ has a single child~$t'$ and~$Q_t = Q_{t'} \setminus \{v\}$.
	We store
	\begin{eqnarray} \label{tw:forget}
		\DP[t,R,G,B,s] =  \max \{
		\DP[t',R,G\cup\{v\},B,s];
		\DP[t',R,G,B\cup\{v\},s]
		\}.
	\end{eqnarray}

	\Recc{tw:forget} follows from the definition that vertices in $v\in V_t \setminus Q_t$ are either black or green with respect to $(t,R,G,B,s)$-feasible sets.

	\proofpara{Join Node}
	Suppose that~$t\in T$ is a \emph{join node}, that is,~$t$ has two children~$t_1$ and~$t_2$ with~$Q_t = Q_{t_1} = Q_{t_2}$.
	We call two partitions $R_1\cup G_1\cup B_1$ and $R_2\cup G_2 \cup B_2$ of $Q_t$ \emph{qualified} for $R\cup G\cup B$ if $R=(R_1\cup R_2) \setminus (G_1 \cup G_2)$ and $G = G_1\cup G_2$ (and consequently $B = B_1 \cap B_2$). See~\Cref{fig:joinColorings}.
	We store
	\begin{eqnarray} \label{tw:join}
		\DP[t,R,G,B,s] &=& \max_{(\pi_1, \pi_2) \in{\cal Q},s'} ~ \DP[t_1,R_1,G_1,B_1,s'] + \DP[t_2,R_2,G_2,B_2,s-s'],
	\end{eqnarray}
	%maximum is taken over all partitions $R_1\cup G_1\cup B_1$ and $R_2\cup G_2 \cup B_2$  that are qualified for $R,G,B$ and all $s' \in [s]$.
	where ${\cal Q}$ is the set of pairs of partitions $\pi_1 = R_1\cup G_1\cup B_1$ and $\pi_2 = R_2\cup G_2 \cup B_2$ that are qualified for $R,G,B$, and $s' \in [s]_0$.

	By~\Cref{fig:joinColorings} we observe that in \Recc{tw:join} we consider the correct combinations of~$R$,~$G$, and~$B$.

	\newcommand{\darkgreen}{green!50!black}
	\begin{figure}
		\begin{center}
			\begin{tabular}{c|c|c|c|}\cline{2-4}
				 & $B_1$ & \color{red}{$R_1$} & \color{\darkgreen}{$G_1$}\\\hline
				\multicolumn{1}{|l|}{$B_2$} & $B$ & \color{red}{$R$} & \color{\darkgreen}{$G$} \\\hline
				\multicolumn{1}{|l|}{\color{red}{$R_2$}} & \color{red}{$R$} & \color{red}{$R$} & \color{\darkgreen}{$G$} \\\hline
				\multicolumn{1}{|l|}{\color{\darkgreen}{$G_2$}} & \color{\darkgreen}{$G$} & \color{\darkgreen}{$G$} & \color{\darkgreen}{$G$} \\\hline
			\end{tabular} 
			
			\caption{This table shows the relationship between the partitions $R_1\cup G_1\cup B_1$, $R_2\cup G_2 \cup B_2$ and $R\cup G\cup B$ in the case of a join node, when $R_1\cup G_1\cup B_1$ and $R_2\cup G_2 \cup B_2$ are qualified for $R\cup G\cup B$. The table shows which of the set $R,G$, or~$B$ an element $v \in Q_t$ will be in, depending on its membership in $R_1,G_1,B_1,R_2,G_2$, and $B_2$. For example if $v \in R_1$ and $v \in B_2$, then $v \in R$.}
			\label{fig:joinColorings}
		\end{center}
	\end{figure}
	
	\proofpara{Running Time}
	The table contains $\Oh(3^\tw \cdot nk)$ entries, as a tree decomposition contains $\Oh(n)$ nodes.
	Each leaf and forget node can be computed in linear time.
	An introduce node can be computed in $\Oh(2^\tw\cdot n)$~time.
	In a join node, considering only the qualified sets, $(\pi_1,\pi_2)$ already define $R$, $G$, and $B$.
	Thus, all join nodes can be computed in $\Oh(9^\tw \cdot nk)$~time,
	which is also the overall running time.
\end{proof}
}
\ifWABIShort
\else
\proofTW
\fi

\section{Discussion}
\label{sec:discussion}
In this paper, we studied the algorithmic complexity \PDD and \sPDD with respect to various parameterizations.
\PDD is \FPT when parameterized with the solution size plus the height of the phylogenetic tree.
Consequently, \PDD is \FPT with respect to~$D$, the threshold of diversity.
\ifJournal
However, both problems, \PDD and \sPDD, are unlikely to admit a kernel of polynomial size.
\fi
Further, unlike some other problems on maximizing phylogenetic diversity~\cite{MAPPD,timePD}, \PDD probably does not admit an \FPT-algorithm with respect to~\Dbar, the acceptable loss of phylogenetic diversity.

We also considered the structure of the food-web.
Among other results, we showed that \PDD remains \NP-hard even if the food-web is a cluster graph but \PDD is \FPT with respect to the number of vertices that need to be removed from the food web to receive a co-cluster.
On the positive side, we proved that \sPDD is \FPT with respect to the treewidth of the food-web and therefore can be solved in polynomial time if the food-web is a tree.

Several interesting questions remain open after our examination of \PDD and \sPDD.
Arguably the most relevant one is whether \PDD is \FPT with respect to~$k$, the size of the solution.
Also, it remains open whether \PDD can be solved in polynomial time if each connected component in the food-web contains at most two vertices.

Clearly, further structural parameterizations can be considered.
We only considered structural parameters which consider the underlying graph.
But parameters which also consider the orientation of edges, such as the longest anti-chain, could give a better view on the structure of the food-web than parameters which only consider the underlying graph.

Liebermann et al.~\cite{lieberman} introduced and analyzed weighted food-webs.
Such a weighted model may provide a more realistic view of a species’ effect on and interaction with other species~\cite{cirtwill}.
Maximizing phylogenetic diversity with respect to a weighted food-web in which one potentially needs to save several prey per predator would be an interesting generalization for our work and has the special case in which one needs to save all prey for each predator.

Recent works consider the maximization of phylogenetic diversity in phylogenetic networks~\cite{WickeFischer2018,bordewichNetworks,MAPPD,MaxNPD} which may provide a more realistic evolutionary model of the considered species. It would be interesting to study these problems also under ecological constraints. Do the resulting problems become much harder than \PDD?
Finally, it has been reported that maximizing phylogenetic diversity is only marginally better than selecting a random set of species when it comes to maximizing the functional diversity of the surviving species~\cite{MPC+18}. The situation could be different, however, when ecological constraints are incorporated. Here, investigating the following two questions seems fruitful: First, do randomly selected viable species sets have a higher functional diversity than randomly selected species? Second, do viable sets with maximal phylogenetic diversity have a higher functional diversity than randomly selected viable sets?

\newpage
\thispagestyle{empty}

\thispagestyle{empty}

\ifWABIShort
\newpage
\setcounter{page}{1}
\setcounter{section}{0}
\renewcommand\thesection{A.\arabic{section}}
\section{Appendix}
\paragraph*{Proof of~\Cref{lem:reduction-rules}}
\proofRREdgeOrginal

\proofRREdgePattern

\proofRRFoodWeb

\proofRRInternalVertex

\paragraph*{Proof of~\Cref{thm:D}}
\proofThmD

\paragraph*{All other proofs}
\proofTopPredator

\proofThmKStars

\proofHStw

\proofTW

\fi

\end{document}